\documentclass[12pt,a4paper]{article}
\usepackage[utf8x]{inputenc}
\usepackage[T1]{fontenc}
\usepackage{mathptmx} 
\usepackage{multirow}
\usepackage{amsmath}
\usepackage{amssymb}
\usepackage{amsfonts}
\usepackage{setspace}

\usepackage{calc} 
\usepackage{enumitem} 

\RequirePackage[authoryear]{natbib}
\RequirePackage[colorlinks,citecolor=blue,urlcolor=blue]{hyperref}
\RequirePackage{graphicx}
\usepackage{bbm, dsfont}
\usepackage{soul}
\usepackage[mathscr]{euscript}
\usepackage{mathrsfs}
\usepackage{booktabs}
\usepackage{url} 
\usepackage{setspace}
\usepackage{multicol}
\usepackage{multirow}
\usepackage{graphicx}
\usepackage{ulem}
\usepackage{amsthm}
\usepackage{rotating} 
\usepackage{mathrsfs}
\usepackage{bbold}
\usepackage{amssymb}

\usepackage[a4paper, margin=2cm]{geometry} 

\usepackage[all]{nowidow} 
\usepackage[protrusion=true,expansion=true]{microtype} 

\usepackage[english]{babel}

\newtheorem{theorem}{Theorem}[section]

\newtheorem{definition}[theorem]{Definition}

\newcommand{\rmd}{{\rm d}}
\newcommand{\bs}{\boldsymbol}
\newcommand{\bm}{\boldsymbol}

\newcommand{\Cor}{\mathrm{Cor}}

\newcommand{\rmi}{\mathrm{i}}

\newcommand{\mbf}{\boldsymbol}
\newtheorem{proposition}{Proposition}[section]

\theoremstyle{plain}

\theoremstyle{remark}

\theoremstyle{empty}
\newtheorem*{duplicate}{Proposition}

\begin{document}
\doublespacing
\begin{center}
    \large
    \textbf{KenCoh: A Ranked-Based Canonical Coherence} \footnote{
    Preprint}
        
   \vspace{0.2cm}
    Mara Sherlin D.P. Talento$^{\text{a}*}$,
    Sarbojit Roy$^{\text{b}*}$, 
    Tania Reyes Vallejo$^{\text{c}**}$,
    Leena A Ibrahim$^{\text{d}**}$,
    Hernando Ombao$^{\text{e}*}$
    
    \vspace{0.01cm}
   
    {
    \singlespacing
    \textit{$^a$marasherlin.talento@kaust.edu.sa, $^b$sarbojit.roy@kaust.edu.sa, $^c$tania.reyesvallejo@kaust.edu.sa, $^d$leena.ibrahim@kaust.edu.sa, $^e$hernando.ombao@kaust.edu.sa}

    $^*$Statistics Program, Computer Electrical and Mathematical Science \& Engineering, King Abdullah University of Science and Technology

    $^{**}$  Bioscience Program, Biological and Environmental Science and Engineering Division, King Abdullah University of Science and Technology
    }

    \vspace{0.9cm}
    
    \textbf{Abstract}
\end{center}
This work is inspired by the problem of characterizing a dependence measure between two cortical regions of the brain where each region contains multiple signal recordings from several neurons or channels (e.g., inhibitory and excitatory neurons). The goal is to identify differences in the structure of brain functional connectivity between known brain states. An exploratory tool for studying the dependence between two random vectors is via canonical correlation analysis. However, these are limited to only capturing linear associations and are sensitive to outlier observations. Mitigating these limitations is crucial because brain functional connectivity is likely to be more complex than linear, and brain signals may exhibit heavy-tailed properties. To overcome these limitations, we develop a robust method, Kendall's tau-based canonical coherence (KenCoh), to learn connectivity structure among neuronal signals filtered at given frequency bands. Our simulation study demonstrates that KenCoh is competitive with the moment-based estimator and outperforms the latter when the underlying distributions are heavy-tailed. We apply our method to EEG recordings from a virtual-reality driving experiment and to calcium imaging recordings in inhibitory and excitatory neurons of the auditory cortex in mice subjected to sound stimuli. Our findings reveal distinct regional dependencies across frequency bands and brain states. 

\textbf{Keywords:} dependence; electroencephalography; robust; multivariate; spectral analysis.
\vfill

\newpage

\section{Introduction}\label{chap:introduction}

Cognitive task-related brain responses are typically linked to the activity of specific brain regions or neuronal populations, as well as specific brain wave patterns, rather than interactions between individual neuron pairs. For instance, 
research has shown that the occipital and parietal regions---responsible for visual processing and spatial awareness---play a key role in collision avoidance while driving \citep{Spiers2007Driving}. 
\cite{Shi2023Frontal} observed heightened activation in the frontal and prefrontal cortical regions in response to visual and auditory cognitive distractions. These findings suggest that certain cortical areas, e.g., the Temporal-Parietal-Occipital (TPO) and Frontal-Prefrontal (Fp) regions, are critically involved in attention and cognition during task-related experiments.

Despite this growing understanding, commonly used metrics for wave-specific associations in time series analysis primarily focus on pairwise relationships between signals, rather than connectivity between groups of time series \citep{Ombao2008evolutionary}. Hence, in this paper, we develop a statistical method that explores variation in connectivity between two brain regions or neuronal populations for different brain conditions (e.g., drowsy and alert). 
We now introduce here the electroencephalography (EEG) dataset from \cite{Cao2019DrivingData} which is publicly available at \url{https://figshare.com/articles/dataset/EEG_driver_drowsiness_dataset/14273687?file=30707285}. It contains multi-channel EEG recordings collected during a virtual-reality driving experiment. 

EEG, a non-invasive technique, involves placing multiple electrodes (or channels) on the scalp at predefined locations to capture brain signals.
\begin{figure}
    \centering
    \includegraphics[width=1.05\textwidth]{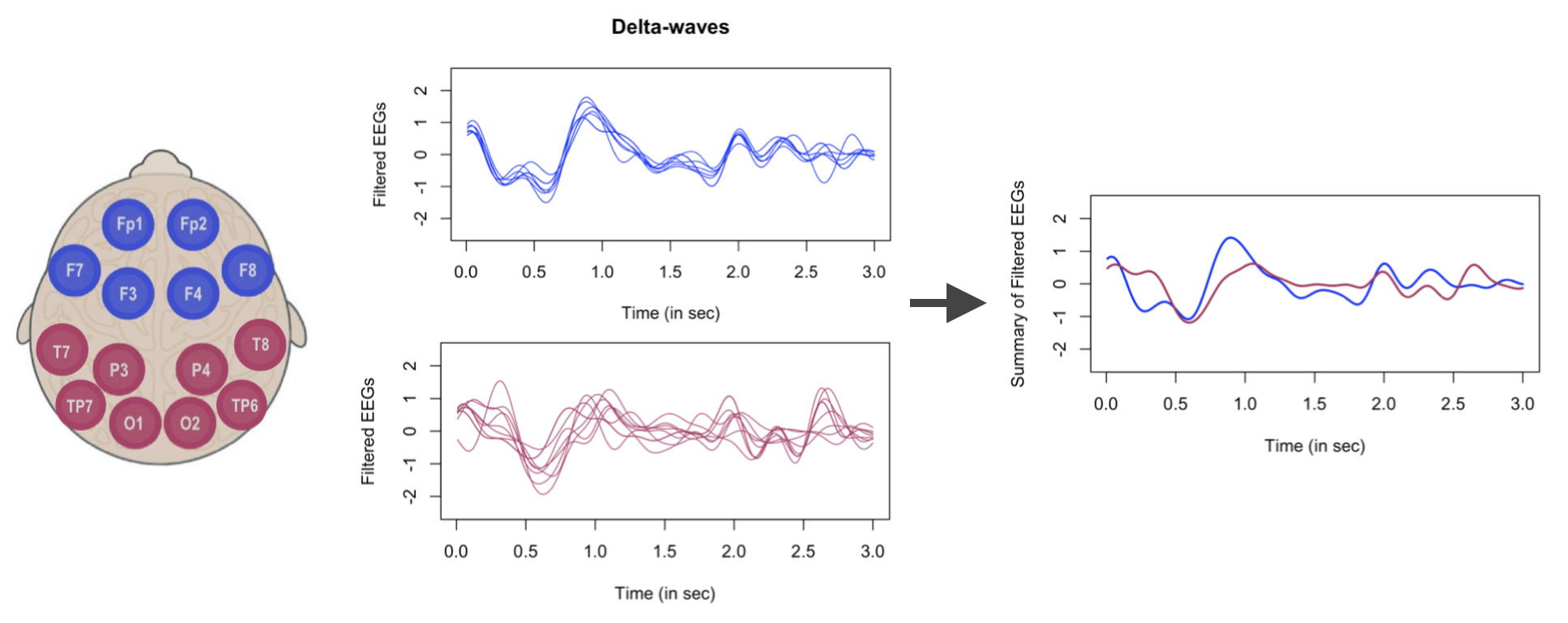}
    \caption{EEG from the driving experiment. (Left) Blue colored nodes are channels on top of the Fp region and red colored nodes are channels on top of the TPO region. (Mid-top) Multiple signals of Fp lobe EEG channels and (Mid-bottom) that of TPO lobe EEG channels filtered at the Delta band (0,4] Hz. (Right) ``Summary'' of filtered signals representing EEG channels of Fp lobe and TPO lobe. }
    \label{framework}
\end{figure} 
Figure~\ref{framework} shows 14 EEG channels located on the scalp that are named after specific brain regions, i.e., blue colored nodes are channels on the anterior regions or the Fp region, namely, Fp1, Fp2, F3, F4, F7, and F8; and red colored nodes are channels on the posterior regions or the TPO region, namely, P3, P4, T7, T8, TP7, TP6, O1, and O2. The mid-panel of Figure~\ref{framework} plots the ``oscillatory'' signals (formally defined in Section~\ref{sec:SpectralAssoc}) extracted from the EEG channels of Fp and TPO regions. 
This plot demonstrates a high degree of synchrony among signals recorded from the same brain region, as reflected in the summarized (representative) signal shown in the rightmost panel of Figure~\ref{framework}.

One disadvantage of analyzing channel pairs one at a time (rather than as a group) is that it
overlooks the broader concurrent dependencies among signals within the same group, which is naturally-occurring in many applications \citep{Krafty2016discriminant}. 
Hence, our goal is to characterize the cross-regional interactions (or dependence) between two regions, while accounting for within-region synchrony. 
This better aligns with the spatial resolution inherent in brain signals, e.g., EEGs. 
There is relatively limited literature addressing correlation between two groups of time series. Most existing methods are extensions of the classical canonical correlation analysis (CCA) originally developed by \cite{HHotelling1936} \citep[also see][]{Hotelling1992CCA}. \cite{He2004CanCorFD} extended CCA to function-valued random vectors, and this framework was further generalized by \cite{Eubank2008CanCorSP}, allowing one or both inputs to be either finite-dimensional vectors or functional data. However, these approaches operate solely in the time domain, whereas many biomedical time series require frequency-domain analysis \citep{Bruce2020empirical}. \cite{Krafty2013CanCorTS} 
examine canonical correlations of a vector-valued variable and an infinite-dimensional representation of time series derived from its spectral density. 
In contrast, our work focuses on developing canonical correlation methods that are explicitly frequency-based.
A more closely related line of work is formulation of canonical correlation for stationary time series in \cite{brillinger1969canonical}. The method proposed in \cite{brillinger1969canonical} maximizes inter-series association at each frequency to construct canonical variates. Building on this, our paper introduces a novel method that estimates monotonic canonical association over predefined frequency bands (defined in Equation~\eqref{FreqBand}). Our proposed Kendall's tau-based canonical coherence (KenCoh) is a rank--based measure for dependence between two groups of signals with the advantage of not reliant on moment existence. Our method is robust to the presence of heavy-tailed or outlier-contaminated signals, addressing a critical limitation of earlier approaches.  

The main goal of this paper is to propose a novel statistical method for estimating overall coherence between multivariate time series (MTS) within specified frequency bands. 
The primary contribution of KenCoh is the development of a robust framework for assessing functional connectivity, along with a formal statistical inference procedure to compare connectivity across cognitive states. The key advantages of KenCoh are (i) a comprehensive measure of connectivity within frequency bands $\Omega$ associated with brain function; (ii) robustness to outliers in time series data; and (iii) the ability to reveal differences in brain networks, such as those observed between alert and drowsy states in a virtual driving task.

The remainder of the article is organized as follows. In Section~\ref{sec:SpectralAssoc}, we discuss the measures of dependence between time series in spectral domain and highlight the limitations of existing approaches. Section~\ref{chap:methodology} presents \textbf{KenCoh}, our proposed method for estimating canonical band-coherence between groups of time series. In Section~\ref{chap:simulation}, we evaluate the robustness of KenCoh through simulation under various settings. In Section 5, we analyze the EEG data and the calcium recording data using the KenCoh method. The KenCoh analysis reveals
novel insights into cognitive brain dynamics.
Finally, Section \ref{chap:conclusion} summarizes the conclusion of the study and limitations of the proposed method.

\section{Spectral dependence} \label{sec:SpectralAssoc}

We discuss here the different spectral dependence measures available in the literature. In Section~\ref{subsec:coh}, we describe coherence which is a pairwise measure of spectral dependence. Section~\ref{subsec:limitation} presents and compares different approaches for defining spectral dependence between two groups of weakly-stationary time series and their limitations. 

\subsection{Coherence between a pair of variables} \label{subsec:coh}
To characterize dependence between group of signals, our focus will be on the synchronization of oscillations from different brain regions. 
Let ${\mbf X}(t) = (X_1(t), \ldots, X_P(t))^\top$ be a collection of EEG signals from one group of channels, which we assume to be weakly stationary over an epoch or time block indexed by time $t = 1,\dots,T$. 
This weakly stationary multi-channel signal can be represented as a superposition of the Fourier waveforms with random amplitudes via the Cram\'er representation \citep{Priestley1967power}. Let 
$\{ e^{\rmi 2 \pi \omega t}, \omega \in (-\frac{1}{2}, \frac{1}{2})\}$ be the Fourier basis functions, where 
$\rmi = \sqrt{-1}$. Then, the Cram\'er representation of $\bs{X}(t)$ is 
\[{\mbf X}(t) = \int_{-0.5}^{0.5} e^{\rmi 2 \pi \omega t} \rmd{\mbf A}(\omega),\]
where the increment process $\rmd{\mbf A}(\omega) \in \mathbb{C}^{P}$ satisfies $\mathbb{E}[\rmd{\mbf A}(\omega)] = \mbf{0}$, $\text{Cov}(\rmd{\mbf A}(\omega), \rmd{\mbf A}(\lambda)) = {\mbf 0}$ when $\lambda \ne \omega$ and $\text{Cov}(\rmd{\mbf A}(\omega), \rmd{\mbf A}(\omega)) = {\mbf f}(\omega) \rmd\omega$,
where ${\mbf f}(\omega)$ is the $P \times P$ spectral matrix which is Hermitian positive semi-definite. The diagonal elements, $f_{jj}(\omega)$, $j = 1, \dots, P$, are the auto-spectra, and the off-diagonal elements, $f_{jk}(\omega)$, $j \neq k = 1, \dots, P$, are the cross-spectra. Since ${\mbf f}(\omega)$ is Hermitian, then $f_{jk}(\omega) = f^*_{kj}(\omega)$, where $f^*_{kj}(\omega)$ is the complex-conjugate of $f_{kj}(\omega)$, for $j \neq k = 1, \dots, P$. 

We now define coherence which is a well-known measure of pairwise synchrony \citep[]{shumway2000time, brillinger2001time}. 
Coherence
between channels $j$ and $k$ at frequency $\omega$ is 
\begin{eqnarray}
    \rho_{jk}(\omega) & = & \|\text{Cor}(\rmd A_j(\omega), \rmd A_k(\omega) \|^2 \notag \\
    & = & \frac{|f_{jk}(\omega)|^2}{f_{jj}(\omega)f_{kk}(\omega)} \in [0,1]. \label{coherence}
\end{eqnarray}
A coherence value close to 1 indicates high synchrony of signals at $\omega$-frequency.
From the above relations, we note that (a.) the diagonal elements of the spectral matrix, ${\mbf f}(\omega)$, captures the relative contribution of all oscillations to the total variability for each channel and (b.) coherence is a frequency-domain analogue of squared cross-correlation.
To provide an intuitive interpretation of coherence,
let us define the oscillatory component at frequency-$\omega$ and channel-$j$ to be ${X_{j,\omega}(t) = e^{\rmi 2 \pi \omega t} \rmd A_j(\omega)}$, for $j = 1, \dots, P$.
Since $|e^{\rmi 2 \pi \omega t}|^2 = 1$, it follows that the squared cross-correlation between these oscillations is ${\|\text{Cor}( X_{j,\omega}(t), X_{k,\omega}(t) ) \|^2 = \|\text{Cor}(\rmd A_j(\omega), \rmd A_k(\omega) ) \|^2}$ which is the coherence in Equation~\eqref{coherence}.
This result implies that coherence is a cross-dependence measure between the oscillatory activity of two channels \citep[see][]{Ombao2008evolutionary, ombao2022spectral}. 

In typical EEG analysis, coherence between signals is computed 
over a specific frequency \textit{band} rather than at a singleton frequency \citep{Bruce2020empirical, ombao2024spectral}. The standard frequency bands in EEG analysis are the Delta $(0,4]$ Hz, Theta $(4,8]$ Hz, Alpha $(8,12]$ Hz, Beta $(12,30]$ Hz and Gamma $(30,50]$ Hz \citep{Srinivasan2007EEG}. 
We now define a frequency band for sampling rate $S$ (i.e., $S$ observations per second), as
\begin{equation}
    \Omega = \{S\omega : \omega \in (\omega_1, \omega_2)\}    \label{FreqBand}
\end{equation}
where $0<\omega_1<\omega_2<1/2$.
In our analyses, coherence between a pair of channels is estimated by applying a bandpass filter.
We filter the series within this band through an absolutely summable $c_{\Omega}(s)$, $s \in \{ 0, \pm1, \pm 2, \ldots \}$.
Its Fourier transform, denoted by $C^{(\Omega)}(\omega) = \sum_{s = -\infty}^{\infty} c_{\Omega}(s) e^{-\rmi 2\pi\omega s}$, satisfies
\begin{equation}
    \left|C^{(\Omega)}(\omega)\right|^2 =\begin{cases}
        1/2\delta \; \; \; \; \text{ for } \omega \in \Omega/S, \\
        0 \; \;\;\; \; \; \; \text{ for } \omega \notin \Omega/S, 
    \end{cases} \label{linearfilter}
\end{equation}
where $\delta = \omega_2 - \omega_1 \in (0, 0.5)$. 
We define the filtered signal at channel $X_j(t)$, denoted by $X_{j, \Omega}(t)$, as ${X_{j, \Omega}(t)  = \sum_{s = -\infty}^\infty c_{\Omega}(s) X_j(t-s)}$, for $j = 1, \dots, P$. This filtered series $X_{j, \Omega}(t)$ has zero power-spectrum outside the frequency band $\Omega$ \citep[see][for details]{Ombao2008evolutionary}. 
The coherence between channels $j$ and $k$ at frequency band $\Omega$ is 
\begin{gather}
    \rho_{jk}(\Omega) = \max_{\ell_0} \frac{|f_{jk}(\Omega, \ell_0)|^2}{f_{jj}(\Omega)f_{kk}(\Omega)}, \; \text{ where } \notag \\
    f_{jk}(\Omega, \ell_0) = \int_{\Omega} f_{jk}(\omega) e^{\rmi 2\pi\omega \ell_0} \rmd\omega \; \; \; \text{and} \; \; \; f_{jj}(\Omega) = \int_{\Omega} f_{jj}(\omega) \rmd\omega,  \label{f-Omega}
\end{gather} 
and $\ell_0 = 0, \pm1, \dots, \pm L$ characterizes lead-lag dependence between $X_{j, \Omega}(t)$ and $X_{k, \Omega}(t)$; i.e., the $\ell_0$ is the time-delay and $2\pi\omega \ell_0$
is the phase-shift in the relationship of $j$-th and $k$-th channels at frequency $\omega$ \citep[for discussion of phase-shift, see Chapter 7 of][]{shumway2000time, ombao2022spectral}. 
The above relation implies that the observed filtered signals can be used to estimate the coherence at each frequency band. 

\subsection{Between-group analysis within a frequency band} \label{subsec:limitation}

We now return to our main interest which is to characterize spectral dependence between two brain regions, or between two groups of channels, denoted as ${{\mbf X}(t) = (X_1(t), \ldots, X_P(t))^\top}$, and ${\mbf Y}(t) = (Y_1(t), \ldots, Y_Q(t))^\top$, for $t = \{1, \dots, T\}$ (as depicted in Figure~\ref{framework}). 
A natural way is to aggregate the signals that belong in the same group through weighted means. Suppose we have $\mbf{u} \in \mathbb{R}^P$ and $\mbf{v} \in \mathbb{R}^Q$. We obtain a summary for $\mbf{X}(t)$ and $\mbf{Y}(t)$, denoted as $\widetilde{X}(t) = \mbf{u}^\top\mbf{X}(t) \in \mathbb{R}$ and $\widetilde{Y}(t) = \mbf{v}^\top\mbf{Y}(t) \in \mathbb{R}$, respectively. 
A naive approach is to assign ``equal weights" to the signals within the same group, that is, $\mbf{u} = (1/P , \ \dots, \ 1/P)$ and $\mbf{v} = (1/Q , \ \dots, \ 1/Q)$. 
This approach, i.e., assigning a pre-determined equal weights to each channel in the group, can be problematic when there are erratic channels or channels which do not contribute to the dependence between the two groups. 

\begin{figure}
    \centering
    \includegraphics[width=0.9\textwidth]{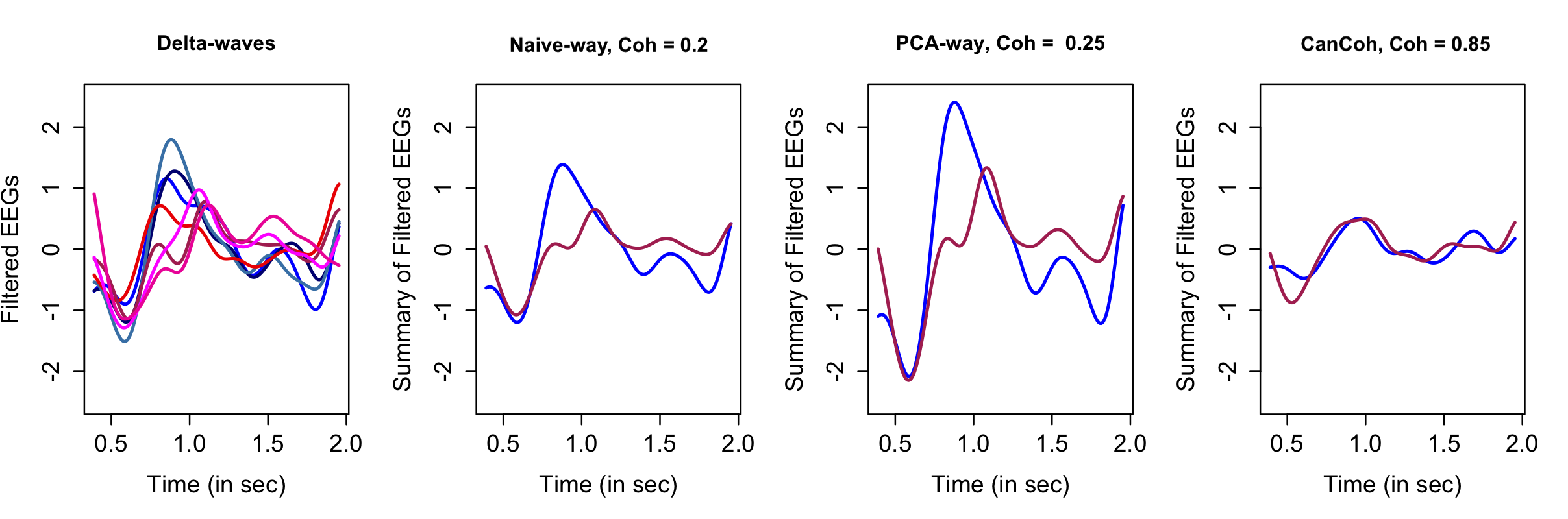}
    \caption{Delta-band filtered EEG signals from the driving dataset in \cite{Cao2019DrivingData} (\textit{showing the first half of the series}). The leftmost panel contains signals from the left Fp lobe (blue) and the left TPO lobe (red). The remaining three panels present regional averages of these signals computed using equal weights (mid-left), PCA-based weights (mid-right), and canonical correlation weights (rightmost).}
    \label{Agg2}
\end{figure} 

An alternative is a data-driven method for determining the weights.
Principal component analysis (PCA) offers a data-driven weight assignment that capture the maximum variation in the data \citep[see][for discussion]{Anderson1958Multivar}, that is, we let $\{\widetilde{X}(t)\}_{t = 1}^T$ and $\{\widetilde{Y}(t)\}_{t = 1}^T$ to be the first principal component scores. 
To illustrate this, consider the EEG signals from the left Fp region (Fp1, F3, F7) and the left TPO region (P3, TP7, T7, O1) using a subset of the data from \cite{Cao2019DrivingData}. Figure~\ref{Agg2} presents regional aggregates computed using the two methods described earlier. The figure highlights clear differences among the aggregation approaches, particularly in the representation of the left TPO signals (red traces) around their peak activity. Both the naive and PCA methods appear to be strongly affected by phase shifts, leading to multiple smaller peaks. 


\paragraph{Limitations of the existing methods. } To evaluate the effectiveness of these approaches in capturing multivariate associations, we present a simple simulation study.
Let $\{\mbf{Y}(t)\}_{t = 1}^T$ be a function of $\{\mbf{X}(t)\}_{t = 1}^T$ and white noise, for $T = 128\text{Hz} \times 3 \text{sec}$. Specifically, let $W_k(t) \sim N(0, 0.0001)$ be the white noise component with small variance, for $k = {1, \dots, 8}$ and $t = 1, \dots, 384$. Moreover, we have ${X_1(t) = 0.85 X_1(t - 1) - 0.73X_1(t-2) + W_1(t)}$, and $X_j(t) = W_j(t)$ for $j = \{2, 3, 4\}$, i.e., $X_1(t)$ is a second-order autoregressive process. 
Finally, we define $Y_j(t) = (j/10) X_j(t), \; \text{for } t = \{ 1, \dots, 384\} \; \text{and }  j = \{1, \dots, 4\}$.
In this simulation, we expect $\mbf{Y}(t)$ to be highly associated with $\mbf{X}(t)$. In fact, we can find $\widetilde{X}(t)$ and $\widetilde{Y}(t)$, 
such that Cor$(\widetilde{X}(t),\widetilde{Y}(t))$ is close to one. 
\begin{figure}
    \centering
    \includegraphics[width=0.9\textwidth]{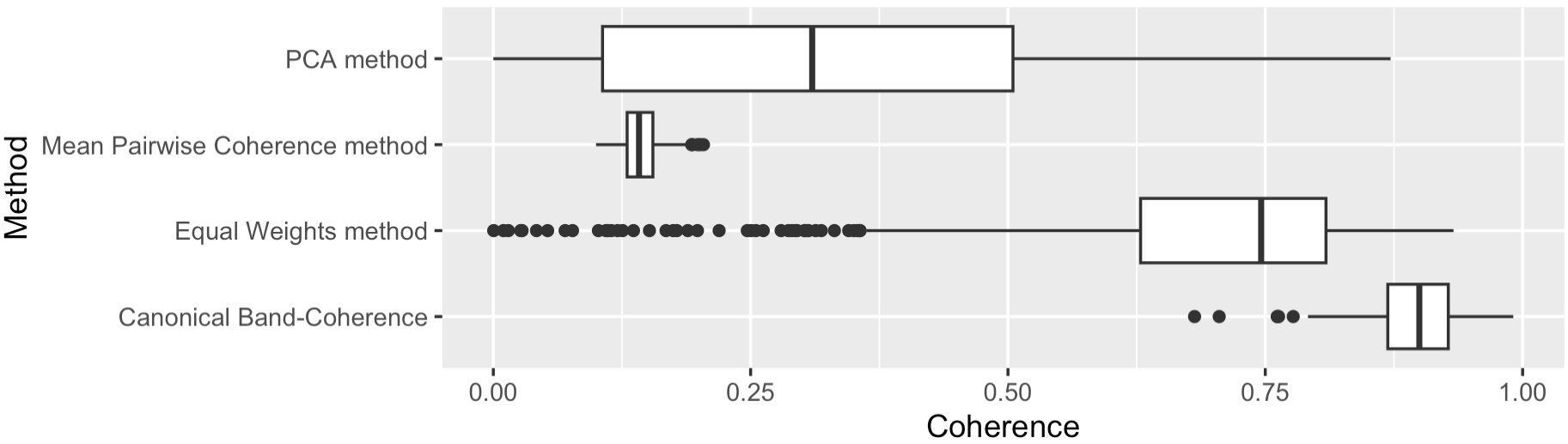}
    \caption{Coherence between $\widetilde{X}$ and $\widetilde{Y}$ at Theta-band of simulated MTS using the four different methods, namely, canonical coherence method, equal weights method, mean of pairwise coherences method and PCA method.}
    \label{OverallCoh}
\end{figure} 

We replicate the experiment 500 times and apply the different approaches described above to estimate the overall coherence between $\widetilde{X}(t)$ and $\widetilde{Y}(t)$. 
The first approach is using the ``equal weights'', the second approach is the ``PCA method'', and the third approach is the novel method for obtaining the ``canonical band-coherence'' that we define in Definition~\ref{definition:CBC}. 
In addition to these methods of aggregation, we also take the mean of all pairwise coherence between ${\mbf X}(t)$ and ${\mbf Y}(t)$, for comparison.

From Figure~\ref{OverallCoh}, we observe that averaging of all pairwise cross-coherence values between $\mbf{X}(t)$ and $\mbf{Y}(t)$ (see second box-plot from the top) results in an underestimated measure of overall coherence. This approach yields a high coherence value only when all components of $\mbf{X}(t)$ and $\mbf{Y}(t)$ are perfectly synchronized. Consequently, even if a subset of components exhibits strong synchronization (possibly with some phase-shift), then the overall measure remains low, thereby failing to provide an effective summary of connectivity between groups of time series. On the other hand, aggregating the time series first and then computing coherence (see third box-plot from the top) yields relatively higher coherence values.
However, this approach suffers from a key limitation---it imposes fixed weights during aggregation (i.e., assigned prior to coherence estimation), offering little interpretability regarding the individual contributions of each component to the coherence. Moreover, although the PCA method offers a data-driven way of assigning weights, it gives a low measure of overall association. This is because PCA does not account for the cross-covariance in its procedure. To mitigate these limitations, we introduce in Section~\ref{subsec2:CBC} a novel, interpretable measure of canonical band-coherence, grounded in the concept of coherence between two MTS at a single frequency \citep{brillinger1969canonical, Lyubushin1998analysis}.

\paragraph{Canonical coherence for weakly stationary time-series.} Define a ${(D = P+Q)}$-dimensional zero-mean weakly stationary time series as ${\mbf{Z}(t) = (X_1(t), \dots, X_P(t), Y_1(t), \dots, Y_Q(t))}$.
The corresponding spectral density matrix of $\mbf{Z}(t)$ is
\[ {\mbf f}_{ZZ}(\omega) = \begin{pmatrix}
    {\mbf f}_{XX}(\omega) & {\mbf f}_{XY}(\omega) \\
    {\mbf f}_{YX}(\omega) & {\mbf f}_{YY}(\omega)
\end{pmatrix},\]
where ${\mbf f}_{XX}(\omega)$ is the $P \times P$ auto-spectral matrix of ${\mbf X}(t)$, ${\mbf f}_{YY}(\omega)$ is the $Q \times Q$ auto-spectral matrix of ${\mbf Y}(t)$, and ${\mbf f}_{XY}(\omega)$ is the $P \times Q$ cross-spectral matrix between ${\mbf X}(t)$ and ${\mbf Y}(t)$.
Given vectors $\mbf{\alpha}, \mbf{\alpha}^{*} \in \mathbb{C}^P$ and $\mbf{\beta}, \mbf{\beta}^{*} \in \mathbb{C}^Q$, where $\mbf{\alpha}^{*}$ and $\mbf{\beta}^{*}$ denote the complex-conjugates of $\mbf{\alpha}$ and $\mbf{\beta}$, respectively, the canonical coherence at frequency $\omega$ in \cite{brillinger1969canonical} is 
 \begin{equation*}
    \theta(\omega) = \max_{\{\mbf{\alpha}, \mbf{\beta}\}} \left|\frac{\mbf{\alpha}^{*\top} {\mbf f}_{XY}(\omega) \mbf{\beta}}{\sqrt{\mbf{\alpha}^{*\top} {\mbf f}_{XX}(\omega) \mbf{\alpha} \mbf{\beta}^{*\top} {\mbf f}_{YY}(\omega) \mbf{\beta}} }\right|^2
     \text{such that, } \; \mbf{\alpha}^{*\top} {\mbf f}_{XX}(\omega) \mbf{\alpha} = \mbf{\beta}^{*\top} \mbf{f}_{YY}(\omega) \mbf{\beta} = 1.
 \end{equation*}

Existing approaches to estimating canonical coherence typically rely on moment-based estimators of spectral matrices \citep[]{brillinger2001time, Dehon2000robustCCA}. A key limitation of these estimators is their sensitivity to outliers, which can significantly distort the results. Furthermore, current methods are restricted to singleton-frequency analyses (rather then frequency-band), which do not align well with the band-specific interpretation of latent oscillations commonly used in fields such as neuroscience.
To mitigate these limitations, we develop a new characterization and estimation for canonical coherence that captures the monotonic dependence structure among the signals for a specific frequency band, $\Omega$.


\section{Methodology}\label{chap:methodology}

We now develop our proposed approach. Section~\ref{subsec2:CBC} defines the canonical band-coherence (CBC) and its equivalent form in terms of spectral density matrix. Section~\ref{subsec2:estimation} provides the estimation procedure of Kendall's tau-based canonical coherence (KenCoh) and Section~\ref{subsec4:inference} details the non-parametric hypothesis test used in the analyses. 

\subsection{Canonical Band-Coherence} \label{subsec2:CBC}
Recall the definition of zero-mean weakly stationary time series $\{\mbf{Z}(t)\}_{t = 1}^T \in \mathbb{R}^D$ defined in Section~\ref{subsec:limitation}.
Let $\{Z_j^{\Omega}(t)\}_{t = 1}^T$ denote the bandpass-filtered $\{Z_j(t)\}_{t = 1}^T$ in the $\Omega$-band for ${j = 1, \dots, D}$. 
For a multivariate Gaussian signals, we denote the lagged cross-dependence matrix among filtered signals, on $\Omega$, denoted $\bs{\Gamma}_Z(\Omega, \ell)$, as 
\begin{equation}
    \Gamma_{jk}(\Omega, \ell) = \Cor(Z_j^{\Omega}(t-\ell), Z_k^{\Omega}(t)) = \Cor(Z_j^{\Omega}(t), Z_k^{\Omega}(t+\ell)), \label{Pjk}
\end{equation}
for all $j, k = 1, \dots, D $ and $\ell = 0, \pm 1, \dots, \pm L$. This is also the coherency matrix in \cite{ombao2022spectral}. We now formally define CBC for a class of meta-elliptic distribution.

\begin{definition}[Canonical Band-Coherence] \label{definition:CBC}
Let $\bs{Z}(t) =  (\bs{X}^\top(t), \bs{Y}^\top(t))^\top = (Z_1(t), \dots, Z_{D}(t))$, for $t = 1, \dots, T$, be a weakly stationary time series and $\bs{Z}^{\Omega}(t)$, the vector of corresponding filtered series at frequency band $\Omega$, follow
a multivariate elliptic distribution, that is, 
  $$\begin{pmatrix}
        \bs{Z}^{\Omega}(t) \\ \bs{Z}^{\Omega}(t+\ell) 
    \end{pmatrix} \sim \mathcal{E} \left(\bs{0},\begin{pmatrix}
        \bs{\Gamma}_{Z}(\Omega, 0) & \bs{\Gamma}_{Z}(\Omega, \ell) \\
        \bs{\Gamma}_{Z}(\Omega, -\ell) & \bs{\Gamma}_{Z}(\Omega, 0)
    \end{pmatrix}, \psi \right), \; \text{ with }$$
    $$\bs{\Gamma}_Z(\Omega, \ell) = \begin{pmatrix}
        \bs{\Gamma}_{XX}(\Omega, \ell) & \bs{\Gamma}_{XY}(\Omega, \ell) \\
        \bs{\Gamma}_{YX}(\Omega, \ell) & \bs{\Gamma}_{YY}(\Omega, \ell)
    \end{pmatrix},$$ 
    where $\psi(\cdot)$ is the generator function of the elliptic distribution \citep[see][]{fang2002meta}, 
    $\ell =0, \pm 1, \dots, \pm L$, 
    and $\Gamma_{jk}$ is as defined in Equation~\eqref{Pjk}, for all $j \neq k = 1, \dots, D$.
    The \textit{canonical band-coherence} (CBC) between ${\bs X}(t) \in \mathbb{R}^P$ and ${\bs Y}(t) \in \mathbb{R}^{Q}$, at $\Omega$, 
    \begin{equation}
        \kappa(\Omega) := \max_{ \{\bs{a}_\Omega, \bs{b}_\Omega, \ell \}} \left\{ \bs{a}^{\top}_\Omega \bs{\Gamma}_{XY}(\Omega, \ell) \bs{b}_\Omega\right\}^2, \label{cancor2}
    \end{equation}
    where $\bs{a}_\Omega \in \mathbb{R}^{P}$, $\bs{b}_\Omega \in \mathbb{R}^{Q}$, and $\bs{a}_\Omega^\top \Gamma_{XX}(\Omega, 0)\bs{a}_\Omega = \bs{b}_\Omega^\top \Gamma_{YY}(\Omega, 0)\bs{b}_\Omega = 1$.
\end{definition}

\noindent \textit{Remark 1}: It is important to note here that ${\mbf \Gamma}_{XY}(\Omega, \ell) = {\mbf \Gamma}_{YX}(\Omega, -\ell)^\top$. 

\noindent \textit{Remark 2:} The vectors $\mbf{a}_{\Omega}$ and $\mbf{b}_{\Omega}$, which is also the unstandardized version of weights in Section~\ref{subsec:limitation}, are referred to as the \textit{canonical directions}. These vectors measure the contributions of each filtered signals to $\kappa_\Omega$ (see Figure~\ref{DeltaConnectivitySubj9} for instance). 

\noindent Proposition~\ref{Theorem:definition} gives an equivalent representation of $\kappa(\Omega)$ in terms of the spectral matrix. 

\begin{proposition} \label{Theorem:definition}
Let $c_{\Omega}(s) \in \mathbb{R}, \ s \in \mathbb{Z}$, be a linear filter, such that, ${\sum_{s = -\infty}^{\infty} |c_{\Omega}(s)| < \infty}$. Given ${\{\mbf{Z}(t)\}_{t = 1}^T}$ be a weakly stationary $D$-dimensional time-series as defined in Definition~\ref{definition:CBC}, we have ${Z}_j^\Omega(t) = \sum_{s = -\infty}^\infty c_\Omega(s) {Z}_j(t-s)$ 
for $j = 1, \dots, D$.
Then, 
\begin{gather}
    \kappa(\Omega) = \max_{ \{\mbf{a}_\Omega , \mbf{b}_\Omega, \ell \}} \left| \mbf{a}^{\top}_\Omega \left[ \int_{\{ \omega: \Omega/S\}} \mbf{f}_{XY}(\omega) e^{\rmi 2\pi\omega \ell} \rmd \omega \right] \mbf{b}_\Omega \right|^2 \;\; \label{cancor-theorem} \\
    \text{such that, } {\mbf a}_\Omega^\top {\mbf f}_{XX}(\Omega) {\mbf a}_\Omega = {\mbf b}_\Omega^\top {\mbf f}_{YY}(\Omega)  {\mbf b}_\Omega = 1, \notag
\end{gather}
where $\mbf{a}_{\Omega} \in \mathbb{R}^P$ and $\mbf{b}_{\Omega} \in \mathbb{R}^Q$ are the canonical directions.
\end{proposition}
The proof of Proposition~\ref{Theorem:definition} starts with expanding Equation~\eqref{cancor2} and expressing it in terms of the covariance of filtered series. The covariance of filtered series is then rewritten in terms of the spectral density matrix and simplification of expression yields to Equation~\eqref{cancor-theorem}.
The complete proof of Proposition~\ref{Theorem:definition} is given in Supplementary Material, Section A-1.1. Proposition~\ref{Theorem:definition} shows that through linear filtering, we can maximize the integral of cross spectral density matrix corresponding to the frequency band of interest. It also shows that by finding the lead-lag that maximizes $\kappa(\Omega)$, the definition takes into account the phase-shift of the oscillation, which in this case is equivalent to $2\pi\omega \ell$.

\begin{proposition} \label{proposition:solution} 
Consider the weakly stationary time series filtered at frequency band $\Omega$, ${(\bs{X}^{\Omega\top}(t), \bs{Y}^{\Omega\top}(t))^\top \in \mathbb{R}^{P+Q}}$, and its lagged cross-dependence matrix $\bs{\Gamma}(\Omega, \ell)$ as in Definition~\ref{definition:CBC}, for $\ell = 0, \pm 1, \dots, \pm L$.
Let ${\Lambda}_j(\Omega, \ell)$ be the $j$-th largest eigenvalue of the following:
    \begin{gather}
    \bs{\Gamma}_{XX}^{-1/2}(\Omega, 0)\bs{\Gamma}_{XY}(\Omega, \ell)\bs{\Gamma}_{YY}^{-1}(\Omega, 0)\bs{\Gamma}_{YX}(\Omega, -\ell)\bs{\Gamma}_{XX}^{-1/2}(\Omega, 0), \; \;  \label{Pxy} \\
    \bs{\Gamma}_{YY}^{-1/2}(\Omega, 0)\bs{\Gamma}_{YX}(\Omega, -\ell)\bs{\Gamma}_{XX}^{-1}(\Omega, 0)\bs{\Gamma}_{XY}(\Omega, \ell)\bs{\Gamma}_{YY}^{-1/2}(\Omega, 0). \label{Pyx} 
\end{gather}
for $j = 1, \dots, \min(P,Q)$. Moreover, for non-zero $\Lambda_{1}(\Omega, \ell)$, we denote the corresponding eigenvectors for Equations~\eqref{Pxy} and \eqref{Pyx} at lag $\ell$, as $\mbf{u}_{\Omega} \in \mathbb{R}^P$ and $\mbf{v}_{\Omega} \in \mathbb{R}^Q$, respectively.
Then, $\kappa_\Omega = \max_{\ell} \{{\Lambda}_{1}(\Omega, \ell)\}$,
    \begin{equation}
       \mbf{a}_{\Omega} := \bs{\Gamma}_{XX}(\Omega, 0)^{-1/2}\mbf{u}_\Omega \ \text{and } \  \mbf{b}_{\Omega} := \bs{\Gamma}_{YY}(\Omega, 0)^{-1/2}\mbf{v}_\Omega. \label{canvec}
    \end{equation}
\end{proposition}

\noindent 
The proof of Proposition~\ref{proposition:solution} is provided in the Supplementary Material, Section A-1.2.
It follows from Equation~\eqref{cancor2} that $\mbf{u}_{\Omega}$ and $\mbf{v}_{\Omega}$ satisfy the condition $\mbf{u}_{\Omega}^\top \mbf{u}_{\Omega} = \mbf{v}_{\Omega}^\top \mbf{v}_{\Omega} = 1$. For fixed $\Omega$, the vectors $\mbf{u}_{\Omega}$ and $\mbf{v}_{\Omega}$ are called the \textit{standardized canonical direction} of filtered $\mbf{X}(t)$ and $\mbf{Y}(t)$, respectively.

An advantage of the above model is the \underline{interpretability} of the parameters. Specifically, note that the eigenvalues and eigenvectors provide different lenses of connectivity. While the largest eigenvalue is equivalent to the maximum band-coherence between two sets of signals (which measures the strength of dependence), the eigenvectors show the relative contribution of the channels to the overall association (see Figure~\ref{DeltaConnectivitySubj9} for illustration).
This paper focuses on the \underline{structure} of connectivity (i.e., the relative contribution of the channels) between the two brain regions attributed to specific band-oscillation, and we develop a statistical procedure to test the difference between canonical directions.

\subsection{KenCoh: Estimation of $\kappa(\Omega)$ using Kendall's tau} \label{subsec2:estimation}

\begin{figure}
    \centering
    \includegraphics[width=1.0\textwidth]{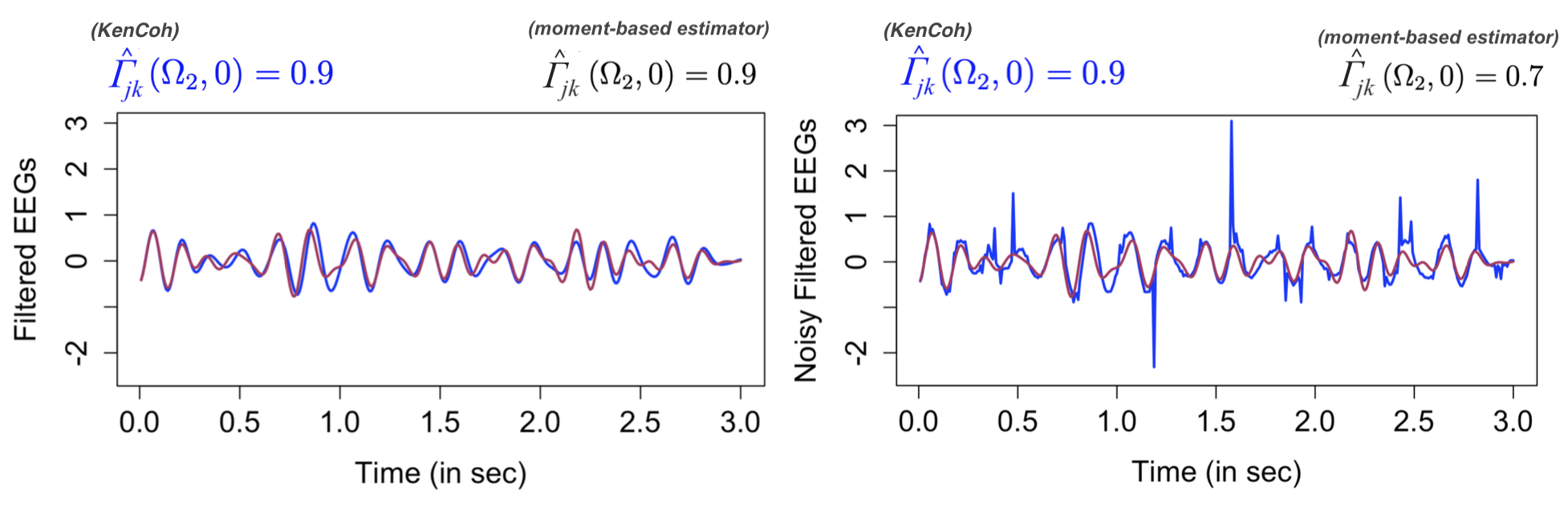}
    \caption{Two EEGs from \cite{Cao2019DrivingData} dataset, filtered at Theta band, $\Omega_2$. The KenCoh statistic at lag $\ell = 0$, i.e., $\hat{\Gamma}^{(1)}_{jk}(\Omega_2, 0)$, for data without contamination (Left) vs with contamination (Right). }
    \label{outliers}
\end{figure} 

Brain signals, including EEGs and calcium recordings, are contaminated with noise. 
The standard estimators of the spectral and covariance matrices are sensitive to outliers. This can produce misleading results about the strength of the connectivity and the brain functional structure captured by the eigenvalue and eigenvectors, respectively.
To handle the outliers, the proposed KenCoh method utilizes concordance of rankings of filtered signals. Here, we develop a \underline{novel method} which includes measures of channel contribution and estimation of dependency between two set of signals. That is, we generalize the concept of the classical canonical coherence by encompassing monotonic dependence and building on the robust approaches for dependence.
We develop a rank-based estimator for CBC for the class of meta-elliptic distribution \citep{fang2002meta, Langworthy2021CanonicalCopula}. 

Let $R^\Omega_j(t,s) = \text{rank}(Z^\Omega_j(t)) - \text{rank}(Z^\Omega_j(s))$. Also, let the number of concordance and discordance, denoted as $q_{jk}(\ell)$ and $q'_{jk}(\ell)$, respectively, to be
\[
q_{jk}(\Omega, \ell) = \sum_{1\leq t<s\leq T} \mathbb{I}_{\{(R^\Omega_j(t,s)R^\Omega_k(t-\ell,s-\ell)) > 0\}}; \; \ q'_{jk}(\Omega, \ell) = \sum_{1\leq t<s\leq T} \mathbb{I}_{\{(R^\Omega_j(t,s)R^\Omega_k(t-\ell,s-\ell)) < 0\}}. 
\]
\noindent where $\mathbb{I}_{\mathcal{B}}$ denotes the indicator function, taking the value $1$ if $\mathcal{B}$ holds and $0$ otherwise.  
We estimate the $j,k$-th element of the matrix $\bs{\Gamma}(\Omega,\ell)$, for $j,k = 1, \dots, D$, as the following:
\begin{equation}
    \hat{\Gamma}^{(1)}_{jk}(\Omega, \ell) = \sin\left(\frac{\pi (q_{jk}(\Omega, \ell) - q'_{jk}(\Omega, \ell))}{2{T \choose 2}}\right). \label{tsksm}
\end{equation}

\noindent Figure~\ref{outliers} shows a sample analysis for $\hat{\Gamma}^{(1)}_{jk}(\Omega_2, \ell)$ between the two sampled EEGs in \cite{Cao2019DrivingData}---i.e., let $Z^{\Omega_2}_j(t)$ and $Z^{\Omega_2}_k(t)$ be the filtered Fp1 and F7, respectively, where $\Omega_2 := \{\Omega: S\omega = 4, S\omega = 8\}$ (see Equation~\eqref{FreqBand}). 
We also compare the result when the signal in FP1 is contaminated, i.e., let $W_j(t) \sim 0.015t_{1}$, and let the contaminated Fp1 signals be $Z^{\Omega_2}_j(t) + W_j(t)$. 
Observe that the estimator uses ranks of two filtered signals, leading to distinct advantages when the data contain outliers or are subject to contamination.

\subsection{Test of equality between brain states} \label{subsec4:inference}

In this section, we develop a formal statistical inference procedure for testing differences in brain functional connectivity structure across various experimental conditions, patient groups, and brain states.
Since the development of KenCoh is primarily motivated by virtual-reality driving experiment \citep{Cao2019DrivingData} and calcium imaging data in \cite{reyesvallejo2023neurogliaform}, one of our goals is to study the difference in brain network connectivity under two known conditions, namely i) alert vs drowsy in the driving data and ii) in-sequence vs oddball in the mice auditory experiment.

This driving experiment consists of multiple trials considered as independent replicates in the analysis. 
Each trial, say $n = 1, \dots, N$, corresponds to a block of multichannel EEG with sampling rate, $S$. A trial, $n$, contains $T$ number of time points. We denote a block of filtered multichannel EEG in trial $n$ as $\{\mbf{Z}^{\Omega}(t)\}_{t \in \mathcal{T}_n}$ where ${\mathcal{T}_n = \{(n-1)T + 1, \dots, (n-1)T + T\}}$, for $n = 1, \dots, N$. Hence, the entire dataset is composed of $TN$ time points. 
For each trial, $n$, we obtain the estimates of first canonical directions denoted by ${\hat{\mbf{a}}}_{\Omega}(n)$ for the first region and ${\hat{\mbf{b}}}_{\Omega}(n)$ for the second region.  
Define $G_n = 0$ if brain state during trial $n$ is $0$; and $G_n = 1$ if the brain state during trial $n$ is $1$. 
Hence, we can partition the $N$ trials into two sets, denoted as $\mathcal{N}_g$, for $g = 0,1$, where $\mathcal{N}_g = \{n : G_n = g \}$ represents the set of trials in brain state $g$. 
Hence, we decompose the trial-specific, $n$, canonical directions to be ${\mbf{B}}_{\Omega,n}^{(g)} = {\mbf{B}}_\Omega^{(1)} G_n +  {\mbf{B}}_\Omega^{(0)} (1 - G_n)$,
where ${\mbf{B}}_\Omega^{(g)} = (a_{1,\Omega}^{(g)} \ , \dots, a_{P,\Omega}^{(g)} \ , {b}_{1,\Omega}^{(g)} \ , \dots, {b}_{Q,\Omega}^{(g)})^\top \in \mathbb{R}^D$, is the vector containing the canonical directions for brain state $g$ (see Equation~\eqref{canvec}).

Recall that the canonical directions defined in Equation~\eqref{canvec} and in Proposition~\ref{Theorem:definition} is a measure of relative contribution of each of the channels (in each group) to the maximal association. We now formulate a test for detecting differences in regional connectivity between the two brain states through the canonical directions. 
We consider the null hypothesis that there are no differences in the connectivity structure (with respect to oscillations at the frequency band $\Omega$) between the state 1 and state 2, i.e., 
$H_0: {\mbf{B}}_\Omega^{(0)} = {\mbf{B}}_\Omega^{(1)}$, against the alternative hypothesis that there is a difference, i.e., $H_1: {\mbf{B}}_\Omega^{(0)} \neq {\mbf{B}}_\Omega^{(1)}$. 
The difference between two vectors, ${\mbf{B}}_\Omega^{(0)}$ and ${\mbf{B}}_\Omega^{(1)}$, is measured using Hotelling's $T^2$ statistics \citep{Hotelling1931generalization}, which we define below. Let $\bar{\mbf{B}}_\Omega^{(g)} = \frac{1}{|\mathcal{N}_g|} \sum_{n \in \mathcal{N}_g} \hat{\mbf{B}}_{\Omega,n}^{(g)}$, for ${g = 0,1}$. 
Let ${\hat{\mbf{J}}_\Omega^{(g)}  = \frac{1}{|\mathcal{N}_g| -1} \sum_{n \in \mathcal{N}_g} \left( \hat{\mbf{B}}_{\Omega,n}^{(g)} - \bar{\mbf{B}}_\Omega^{(g)} \right) \left(\hat{\mbf{B}}_{\Omega,n}^{(g)} - \bar{\mbf{B}}_\Omega^{(g)} \right)^\top}$.
The test statistic is $$K_\Omega = (\bar{\mbf{B}}_0^\Omega - \bar{\mbf{B}}_1^\Omega)^\top \left(\frac{\hat{\mbf{J}}_\Omega^{(0)}}{|\mathcal{N}_0|} + \frac{\hat{\mbf{J}}_\Omega^{(1)}}{|\mathcal{N}_1|}\right)^{-1} (\bar{\mbf{B}}_0^\Omega - \bar{\mbf{B}}_1^\Omega).$$

The empirical distribution of $K_\Omega$, denoted by $\tilde{\mathbb{F}}_0$, under $H_0$ is approximated through permutation testing. 
Specifically, the group assignments of the trials are randomly permuted 5000 times, and $K_\Omega$ is computed for each permutation. This provides an empirical distribution for $K_\Omega$. 
Denote $p^\Omega = \tilde{\mathbb{F}}_0(K_\Omega)$ as the $p$-value of the observed test statistic $K_\Omega$. 
We reject the null hypothesis when $p^\Omega < \alpha^*$, where $\alpha^*$ is the adjusted level of significance for multiple tests \citep{benjamini2000adaptive} on different frequency bands, $\Omega$.


\section{Simulation Study}\label{chap:simulation}

In order to represent the five mutually-independent frequency bands in Section~\ref{sec:SpectralAssoc}, we use AR(2)-mixture model in \cite{ombao2022spectral}. 
In Section~\ref{subsec:CBC-AR2}, we present a multivariate spectral model with defined CBC. Section~\ref{subsec:simsettings} enumerates the various parameter settings for the simulation studies.  Section~\ref{subsec:comparison} discusses two other estimation methods for comparison. Lastly, Section~\ref{subsec:simresults} presents the results of simulation.

\subsection{Canonical coherence under AR(2) Mixture Model} \label{subsec:CBC-AR2}
Here, EEGs are reasonably modeled as superpositions of oscillations with random amplitudes. Second-order autoregressive processes or AR(2), whose polynomial operator admits complex-valued roots represent oscillatory activity, 
are used in this article as building blocks to simulate data \citep[]{gao2020evolutionary, granados2022brain}. 

We represent the jointly weakly-stationary time series, $\mbf{Z}(t) = (\mbf{X}(t)^\top, \mbf{Y}(t)^\top)^\top \in \mathbb{R}^D$, using a mixture of AR(2) processes.
Let $\mbf{W}(t) = (\mbf{W}^\top_X(t), \mbf{W}^\top_Y(t))^\top \in \mathbb{R}^D$ be a white noise where $\mathbb{E}[\bs{W}(t)] = \bs{0}_{D \times 1}$ and $\mathbb{V}[\bs{W}(t)] = \sigma^2_{W}\bs{I}_{D \times D}$, such that $\bs{I}_{D \times D}$ is the identity matrix, $\bs{0}_{D \times 1}$ is the zero vector and $\sigma^2_{W} \in (0,\infty]$, for all $t \in \mathcal{T}_n$, and $n = 1, \dots, N$. 
Define ${\mbf{O}(t) \in \mathbb{R}^{5}}$ as vector of Gaussian AR(2), i.e., 
${O_{m}(t) = \phi^{(m)}_1 O_{m}(t-1) +  \phi_2 O_{m}(t-2) + \xi_{m}(t)},$ for ${m = 1, \dots, 5}$,
where $\xi_{m}(t) \overset{IID}{\sim} N(0,\sigma^2_{m})$, 
${\phi_2 = -\exp(-2M)}$, ${\phi^{(m)}_1 = 2\exp(-M)\cos(2\pi \zeta_m)}$, and $\zeta_m \in (-0.5, 0.5)$ \citep{ombao2022spectral}. In this paper, we set $M = 1.05$, $\sigma_m = 0.5$ and ${\zeta_m = \{\zeta : S\zeta \in \{2, 6, 10, 20, 40\}\}}$ which corresponds to the five frequency bands. Here, we index the five frequency bands, i.e., $\Omega_m$, as, $m = 1$ for the Delta band, $m = 2$ for the Theta band, $m = 3$ for the Alpha band, $m = 4$ for the Beta band, and $m = 5$ for the Gamma band. 
To ensure that $O_{m}(t)$ has zero power outside $\Omega_m$, for $m = 1, \dots, 5$, we applied a Butterworth filter \citep[see][for discussion]{Redondo2023}. 

The $\mbf{X}(t)$ and $\mbf{Y}(t)$ share the same latent process $\mbf{O}(t)$ that are scaled by the mixing matrices $\mbf{E}_X(t) \in \mathbb{R}^{P \times 5}$ and $\mbf{E}_Y(t) \in \mathbb{R}^{Q \times 5}$.
We now consider the following model:
\begin{equation}
    \mbf{Z}(t) = 
    \begin{pmatrix}
        \mbf{X}(t) \\
        \mbf{Y}(t) 
    \end{pmatrix} = \begin{pmatrix}
         \mbf{E}_X(t) \ \mbf{O}(t) + \mbf{W}_X(t)  \\
        \mbf{E}_Y(t) \ h(\mbf{O}(t)) + \mbf{W}_Y(t) 
    \end{pmatrix} \; \; \; \text{for } t \in \mathcal{T}_n,   \label{rawsim}
\end{equation}

\noindent where $h(O_1, \dots, O_5) = (h_0(O_1), \dots, h_0(O_5))$ is a function with $h_0 : \mathbb{R} \rightarrow \mathbb{R}$.
Observe that both $\mbf{X}(t)$ and $\mbf{Y}(t)$ depend on the latent process $\mbf{O}(t)$, and the relation 
is dictated by the function $h : \mathbb{R}^5 \rightarrow \mathbb{R}^5$.
Moreover, the mixing matrices, $\mbf{E}_X(t)$ and $\mbf{E}_Y(t)$, are considered to be fixed within a trial $n$, that is, ${\{\mbf{E}_X(t)\}_{t \in \mathcal{T}_n} := \mbf{E}_X(n)}$ and ${\{\mbf{E}_Y(t)\}_{t \in \mathcal{T}_n} := \mbf{E}_Y(n)}$. To simulate two distinct groups of MTS, we denote the common matrices as
\begin{equation}
    \{(\mbf{E}_X(n)^\top, \mbf{E}_Y(n)^\top)^\top\}_{n \in \mathcal{N}_g} := \mbf{E}^{(g)} \in \mathbb{R}^{D\times 5}, \ \; \text{for } g = 0,1. \label{mixmat}
\end{equation}
Denote ${E}_{jm}^{(g)}$ as the elements of the $\mbf{E}^{(g)}$ matrix, for $g = 0,1$, $j = 1,\dots, D$ and ${m = 1, \dots, 5}$. If $h(\bs{x}) = \bs{x} \in \mathbb{R}^5$, i.e., $\bs{X}(t)$ and $\bs{Y}(t)$ have linear relationship, then we have a closed form solution for the band-specific spectral density matrix (see Equation~\eqref{f-Omega}). Specifically, we have the analytical expressions of canonical coherence and canonical directions, ${\mbf{B}}_\Omega^{(g)}$, for frequency band $\Omega$ and group $g$. The analytical solution for Equation~\ref{cancor-theorem} is given in Supplementary Material, Section A-3. 

Observe that the analytical solutions are functions of the $\mbf{E}^{(g)}$ matrices and variance of the oscillation. 
Therefore, we may have varying settings for $\mbf{W}(t)$, and still retrieve the analytical expressions of canonical coherence for linear relationship between $\mbf{X}(t)$ and $\mbf{Y}(t)$. A heavy-tailed white noise, $W_j(t)$, mimics a contamination in the sampling of signal or other irremovable artifacts from EEG recording. In the next subsection, we provide a measure of distance for 
the differences in connectivity between the two brain states. 

\subsection{Simulation Settings} \label{subsec:simsettings}

Our simulation model in Equation~\eqref{rawsim} has the following components: (i) the mixing-matrices defined in Equation~\eqref{mixmat}, i.e., $\mbf{E}^{(g)} \in \mathbb{R}^{D \times 5}$, for $g = 0,1$, (ii) the five oscillatory signals, i.e., $\bs{O}(t) \in \mathbb{R}^{5}$, (iii) the function $h(\cdot) : \mathbb{R}^{5} \rightarrow \mathbb{R}^{5}$, and (iv) the white-noise component, i.e., $\bs{W}(t) \in \mathbb{R}^{D}$. Throughout the simulation study, we
have the same process for $\bs{O}(t)$ in all simulated cases, see Section~\ref{subsec:CBC-AR2}. The remaining components are varied as follows. 

We begin with white noise components under three settings---(a) ${\{W_j(t)\}_{t = 1}^T \stackrel{IID}{\sim} N(0,1)}$; (b) zero mean $\{W_j(t)\}_{t = 1}^T \stackrel{IID}{\sim} 0.13t_{(3)}$, for $j = 1, \dots, D$; and (c) $\{W_j(t)\}_{t = 1}^T \stackrel{IID}{\sim} N(0,1)$, for $j = 1, \dots, P$ and zero-mean $\{W_j(t)\}_{t = 1}^T \stackrel{IID}{\sim} 0.1t_{(1)}$ for ${j = P+1, \dots, D}$. The first scenario is the Gaussian case, while the other two are the heavy-tailed cases. In particular, the moments of $\mbf{X}(t)$ and $\mbf{Y}(t)$ in the third case does not exist due to the added noise $\bs{W}(t)$. 
 
\begin{figure} 
    \centering
    \includegraphics[width=0.75\textwidth]{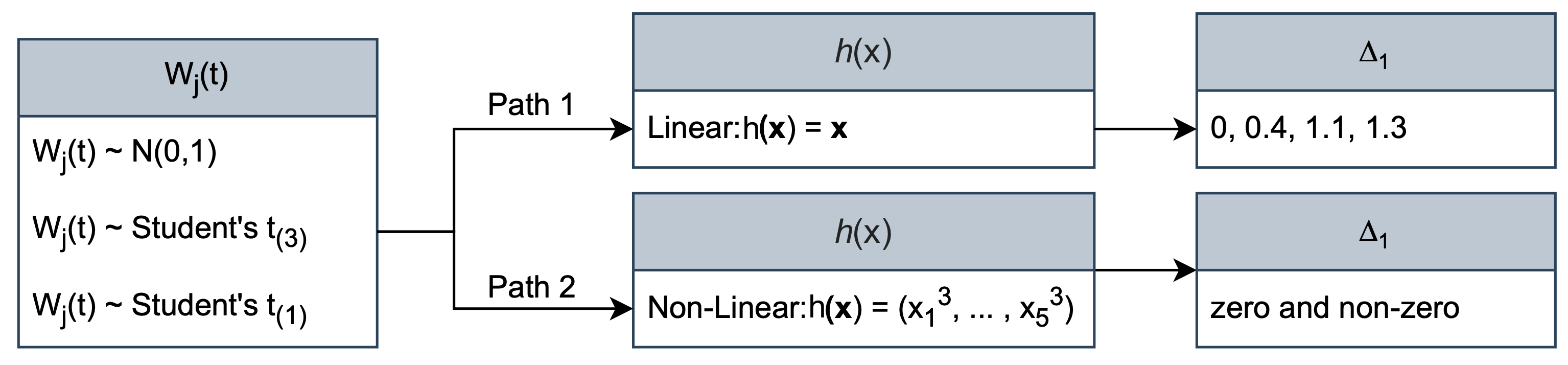}
    \includegraphics[width=0.8\textwidth]{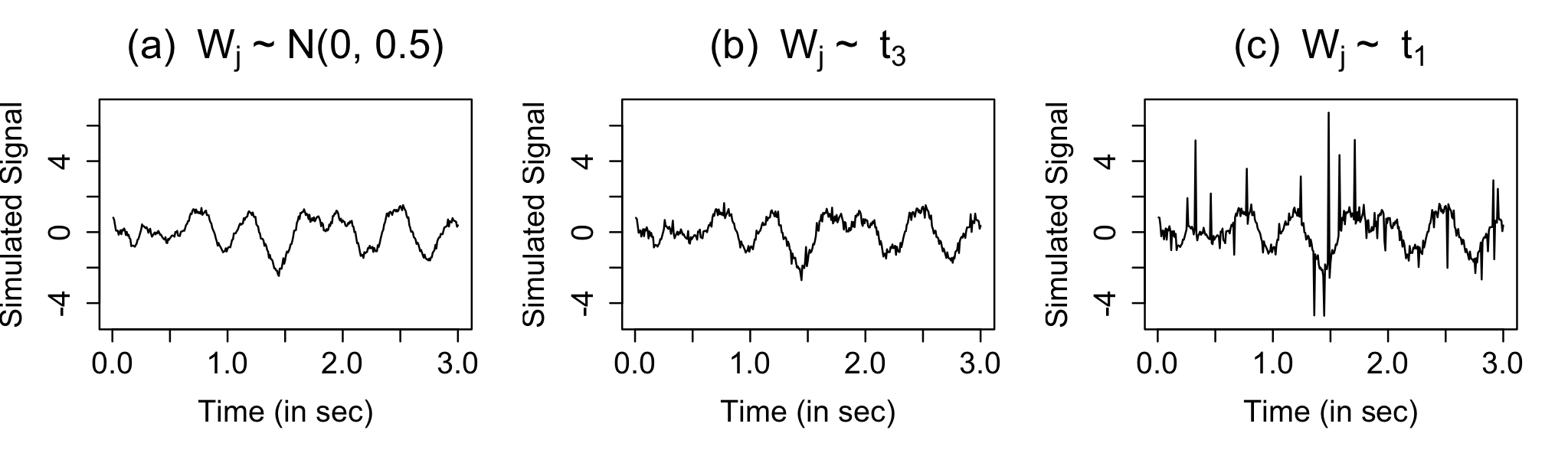}
    \caption{(Top) Simulation settings considered for the numerical experiments. (Bottom) Sample simulated signal, $Z_j(t)$, given the three settings of $W_j(t)$, for a fixed $j$. }
    \label{SimSignals}  
\end{figure}

Next, we let $h(\cdot): \mathbb{R}^5 \rightarrow \mathbb{R}^5$ in Equation \eqref{rawsim} to be either (i) linear, i.e., $h(\bs{x}) = \bs{x}$, or (ii) non-linear, i.e., $h(\bs{x}) = (x_1^3, \dots, x_5^3)^\top$, for $\bs{x} \in \mathbb{R}^5$. Aside from this, we vary the elements of mixing-matrices $\mbf{E}^{(g)}$, for $g = 0,1$. 
Denote $\Delta_{m}$ to be the distance the analytical value of canonical directions between the two states for oscillation $m$, that is,
\begin{equation}
    \Delta_{m} = \frac{||\mbf{u}^{(0)}_{{m}} - \mbf{u}^{(1)}_{{m} }||_{L_2}}{\sqrt{P}} + \frac{||\mbf{v}^{(0)}_{{m}} - \mbf{v}^{(1)}_{{m}}||_{L_2}}{\sqrt{Q}}, \ \; \text{for } m = 1, \dots, 5, \label{deltam}
\end{equation}
where $||\cdot||_{L_2}$ is the $L_2$-norm, $\mbf{u}^{(g)}_{{m}} \in \mathbb{R}^P$ and $\mbf{v}^{(g)}_{{m}} \in \mathbb{R}^Q$ are the standardized canonical directions in Equation~\ref{canvec}, such that $\omega \in \Omega_m/S$, for $g = 0,1$.
The $m$-th column of $\mbf{E}^{(g)} \in \mathbb{R}^{D\times 5}$ is then constructed, such that, the standardized canonical directions of $\{\mbf{Z}(t)\}_{t \in \{\mathcal{T}_n | n \in \mathcal{N}_0\}}$ and $\{\mbf{Z}(t)\}_{t \in \{\mathcal{T}_n | n \in \mathcal{N}_1\}}$ has distance $\Delta_{m}$.
In our numerical experiments, we explored the difference between the two states in the Delta band, i.e., $\Delta_1 \in \{0, 0.4, 1.1, 1.3\}$ while the differences in the other frequency bands were held fixed. 

For each settings defined in Figure~\ref{SimSignals}, 100 $D$-dimensional time series were generated using Equation~\eqref{rawsim} and the test of hypothesis in Section~\ref{subsec4:inference} was applied. 
We aggregated the $p$-values of the replicated tests using the Fisher's method \citep{Yoon2021powerful, Fisher1970statistical}.  
Let $p_r \sim \text{Uniform}(0,1)$ be the $p$-value for the $r$-th test replicate.
Under $H_0$, we have 
\begin{equation}
    \widetilde{p} = -2\sum_{r=1}^{100}\log(p_r) \sim \chi^2_{200}. \label{eq:ptilde}
\end{equation}
Note that $\widetilde{p}$ is large when $\{p_r\}_{r = 1}^{100}$ are dominated by small values, hence, the larger the value of $\widetilde{p}$, the higher the power of rejecting the null.
We then reported $\widetilde{p}$ along with its probability of being large, i.e., $1 - \mathbb{F}(\widetilde{p})$. If $1 - \mathbb{F}(\widetilde{p}) < 0.05$, we reject the null.
Following the setting of data from \cite{Cao2019DrivingData}, the sampling rate was set to 128Hz for a three-second non-overlapping blocks of MTS. 
A single block, then, consists of 384 time points. We set $N = 300$ with 50\% for $\mathcal{N}_0$ and another half for $\mathcal{N}_1$, yielding a data with dimension of $115,200 \times D$ where $D = 8$ (i.e, $p=q=4$). 

\subsection{Other Estimation Methods} \label{subsec:comparison}

For baseline comparison, we present two alternative estimators for the elements of $\bs{\Gamma}(\Omega,\ell)$ in Equation~\eqref{Pjk}---(a) based on Pearson's correlation estimator that we call VCov hereafter, and (b) based on the minimum covariance determinant (MCD) estimator.
The estimator based on Pearson's correlation coefficient is
\begin{equation}
    \hat{\Gamma}^{(2)}_{jk}(\Omega, \ell) = \frac{\hat{\mathbb{C}}_{jk}(\Omega, \ell)}{\sqrt{\hat{\mathbb{C}}_{jj}(\Omega, 0),\hat{\mathbb{C}}_{kk}(\Omega, 0)}}, \label{mle}
\end{equation}
where $\hat{\mathbb{C}}_{jk}(\Omega, \ell) = \frac{1}{T}\sum_{\forall t} (Z^{\Omega}_j(t) - \bar{Z}_j^{\Omega})(Z^{\Omega}_k(t-\ell) - \bar{Z}_k^{\Omega})$ and $\bar{Z}_j^{\Omega} = T^{-1}\sum_{\forall t} Z^{\Omega}_j(t)$ for $j,k = 1, \dots, D$.  
We also compare the performance of our proposed estimation method with a well-known robust covariance estimator, i.e., the Minimum Covariance Determinant (MCD) estimator \citep[see][]{hubert2018minimum}. This is implemented using the \texttt{robustbase} R-package. Details on the use of MCD estimator for $\bs{\Gamma}_Z(\Omega, \ell)$ is also included in the Supplementary Material, Section A-2. 
This robust estimator can be further improved for time series application, e.g., formulating statistics based on \cite{Kim2011robust}. However, this is beyond the scope of this paper.
Our simulation study provides the performance of these estimators, namely, KenCoh (in Equation~\eqref{tsksm}), VCov (in Equation~\eqref{mle}), and MCD (in Supplementary Material, Section A-2), in terms of their ability to detect differences in connectivity between two clusters of MTS. Provided the distribution is light-tailed, Pearson's correlation coefficient is a good candidate for estimation of $\bs{\Gamma}_Z(\Omega,\ell)$. However, if there are outliers in the data, this estimator is expected to perform poorly. While both VCov and MCD are well-suited for detecting linear associations, they are limited in their ability to capture nonlinear or more complex dependencies among variables, see Table~\ref{tab:NonLinPtilde}.

\subsection{Simulation Results} \label{subsec:simresults}

\begin{table}[]
\caption{Summary of $p$-values for null hypothesis that the canonical directions (weights of the channels) are equal for the two states at the Delta-band given that relationship between $\bs{X}(t)$ and $\bs{Y}(t)$ is \underline{linear}. Significant $p$-values are marked with $^{*}$ if it is below 10\%, $^{**}$ if it is below 5\%, and $^{***}$ if it is below 1\%. }
\label{tab:LinPtilde}
\resizebox{\textwidth}{!}{%
\begin{tabular}{@{}ccrrrrrrrrr@{}}
\toprule
 &
   &
  \multicolumn{3}{c}{\textbf{Case-1}} &
  \multicolumn{3}{c}{\textbf{Case-2}} &
  \multicolumn{3}{c}{\textbf{Case-3}} \\ \cmidrule(l){3-11} 
\multirow{-2}{*}{$\Delta_1$} &
  \multirow{-2}{*}{\textbf{Statistic}} &
  \multicolumn{1}{r}{VCov} &
  KenCoh &
  MCD &
  VCov &
  KenCoh &
  MCD &
  VCov &
  KenCoh &
  MCD \\ \midrule
 &
  { $\widetilde{p}$} &
  { 176.11} &
  \multicolumn{1}{r}{{ 146.36}} &
  \multicolumn{1}{r}{171.72} &
  \multicolumn{1}{r}{145.30} & 
  \multicolumn{1}{r}{173.81} & 
  \multicolumn{1}{r}{155.74} & 
  \multicolumn{1}{r}{141.13} &
  \multicolumn{1}{r}{140.60} &
  \multicolumn{1}{r}{158.57} \\ \cmidrule(l){2-11} 
\multirow{-2}{*}{\textbf{0}} &
  $1 - \mathbb{F}(\widetilde{p})$ &
  \multicolumn{1}{r}{0.8872} &
  0.9983 &
  0.9269 &
  0.9986 & 
  0.9095  & 
  0.9910 & 
  0.9994 &
  0.9995 &
  0.9861 \\ \midrule
 &
  $\widetilde{p}$ &
  214.71 &
  \multicolumn{1}{r}{226.10} &
  \multicolumn{1}{r}{167.35} &
  \multicolumn{1}{r}{230.47} & 
  \multicolumn{1}{r}{264.01} & 
  \multicolumn{1}{r}{185.67} & 
  \multicolumn{1}{r}{212.32} &
  \multicolumn{1}{r}{271.66} &
  \multicolumn{1}{r}{194.17} \\ \cmidrule(l){2-11} 
\multirow{-2}{*}{\textbf{0.4}} &
  $1 - \mathbb{F}(\widetilde{p})$ &
  \multicolumn{1}{r}{0.2262} &
  0.0994$^{*}$ &
  0.9551 & 
  0.0687$^{*}$ & 
  0.0016$^{**}$ & 
  0.7584 & 
  \multicolumn{1}{r}{0.2622} &
  0.0006$^{**}$ &
  0.6029 \\ \midrule 
 &
  $\widetilde{p}$ &
  1156.97 &
  \multicolumn{1}{r}{1077.24} &
  \multicolumn{1}{r}{350.57} &
  \multicolumn{1}{r}{1419.53} & 
  \multicolumn{1}{r}{1221.36} & 
  \multicolumn{1}{r}{296.20} & 
  262.24 &
  \multicolumn{1}{r}{356.87} &
  \multicolumn{1}{r}{232.92} \\ \cmidrule(l){2-11} 
\multirow{-2}{*}{\textbf{1.1}} &
  \textbf{$1 - \mathbb{F}(\widetilde{p})$} &
  \textless{}0.0001$^{***}$ &
  \textless{}0.0001$^{***}$ &
  { \textless{}0.0001$^{***}$} &
  { \textless{}0.0001$^{***}$} &
  { \textless{}0.0001$^{***}$} &
  { \textless{}0.0001$^{***}$} &
  0.0021$^{**}$ &
  { \textless{}0.0001$^{***}$} &
  0.0552$^{*}$ \\ \midrule
 &
  $\widetilde{p}$ &
  1260.56 &
  \multicolumn{1}{r}{1134.45} &
  \multicolumn{1}{r}{350.67} &
  \multicolumn{1}{r}{1954.73} & 
  1603.52 & 
  \multicolumn{1}{r}{277.04} & 
  \multicolumn{1}{r}{514.48} &
  \multicolumn{1}{r}{586.21} &
  \multicolumn{1}{r}{287.10} \\ \cmidrule(l){2-11} 
\multirow{-2}{*}{\textbf{1.3}} &
  \textbf{$1 - \mathbb{F}(\widetilde{p})$} &
  \textless{}0.0001$^{***}$ &
  \textless{}0.0001$^{***}$ &
  { \textless{}0.0001$^{***}$} &
  { \textless{}0.0001$^{***}$} & 
  { \textless{}0.0001$^{***}$} & 
  { 0.0003$^{***}$} & 
  { \textless{}0.0001$^{***}$} &
  { \textless{}0.0001$^{***}$} &
  0.0001$^{***}$ \\ \bottomrule
\end{tabular}%

}
\end{table}

In this section, we summarize the results of the power analysis for the examples introduced in Section~\ref{subsec:simsettings}.
Table~\ref{tab:LinPtilde} provides the power of the test, in terms of $\widetilde{p}$ (see Equation~\eqref{eq:ptilde}), for the different settings of $\Delta_1$ given that relationship between $\bs{X}(t)$ and $\bs{Y}(t)$ is linear (i.e., Path 1 in Figure~\ref{SimSignals}, top panel). We expect $\widetilde{p}$ to be high when $\Delta_1>0$ and low when $\Delta_1 = 0$.
Table~\ref{tab:LinPtilde} also shows the $p$-values for $\widetilde{p}$, i.e., $1 - \mathbb{F}(\widetilde{p})$. 
From Table~\ref{tab:LinPtilde}, we observe that the tests based on all three estimation methods yield low values of $\widetilde{p}$ when there is no difference between groups (i.e., $\Delta_1 = 0$). 
In Case 1, where $\mbf{W}(t)$ follows a Gaussian distribution, all tests show increasing power as $\Delta_1$ increases.
However, both the VCov and KenCoh methods consistently demonstrate higher power compared to MCD.
A similar pattern is seen in Case 2, where the added noise $W_j(t)$ follows a scaled Student's $t_3$ distribution. 
In Case 3, where white noise with a Cauchy distribution is added, KenCoh clearly outperforms both VCov and MCD. 
For example, under Case 3, both VCov and MCD fail to reject the null hypothesis for $\Delta_1 = 0.4$, whereas KenCoh successfully rejects for all $\Delta_1 \geq 0.4$, i.e., KenCoh is more sensitive even at smaller values of $\Delta_1$.
The MCD estimator performs poorly across all cases, which may be attributed to its limitation in handling sequential (longitudinal) data. Specifically, it minimizes the covariance determinant under the assumption of independent observations, which does not hold in longitudinal settings \citep{hubert2018minimum}.
In general, VCov performs well when the underlying distribution is elliptical with finite variance, and KenCoh yields competitive results across all cases. However, when the noise distribution lacks a second-order moment, such as when  ${W}_j(t) \sim 0.1t_1$ for all $j = P+1, \dots, D$, KenCoh significantly outperforms VCov. In such heavy-tailed settings, VCov tends to have a higher false negative rate and reduced power in detecting group differences.

\begin{table}[]
\caption{Summary of $p$-values for the null hypothesis that the canonical directions (weights of the channels) are equal for the two states at the Delta-band given that relationship between $\bs{X}(t)$ and $\bs{Y}(t)$ is \underline{non-linear}. Significant $p$-values are marked with $^{*}$ if it is below 10\%, $^{**}$ if it is below 5\%, and $^{***}$ if it is below 1\%.}
\label{tab:NonLinPtilde}
\resizebox{\textwidth}{!}{%
\begin{tabular}{@{}ccrrrrrrrrr@{}}
\toprule
 &
   &
  \multicolumn{3}{c}{\textbf{Case-1}} &
  \multicolumn{3}{c}{\textbf{Case-2}} &
  \multicolumn{3}{c}{\textbf{Case-3}} \\ \cmidrule(l){3-11} 
\multirow{-2}{*}{$\Delta_1$} &
  \multirow{-2}{*}{\textbf{Statistic}} &
  VCov &
  KenCoh &
  MCD &
  VCov &
  KenCoh &
  MCD &
  VCov &
  KenCoh &
  MCD \\ \midrule
 &
  {  $\widetilde{p}$} &
  {  109.73} &
  {  140.40} &
  128.11 &
  116.98 &
  148.31 &
  124.92 &
  133.87 &
  138.33 &
  165.20 \\ \cmidrule(l){2-11} 
\multirow{-2}{*}{\textbf{Zero}} &
  $1 - \mathbb{F}(\widetilde{p})$ &
  1.0 &
  0.9995 &
  1.0 &
  1.0 &
  0.9976 &
  1.0 &
  0.9999 &
  0.9997 &
  0.9655 \\ \midrule
 &
  $\widetilde{p}$ &
  166.93 &
  285.08 &
  135.90 &
  171.55 &
  311.69 &
  154.59 &
  231.96 &
  237.96 &
  147.17 \\ \cmidrule(l){2-11} 
\multirow{-2}{*}{\textbf{Non-zero}} &
  $1 - \mathbb{F}(\widetilde{p})$ &
  0.9574 &
  0.0001$^{**}$ &
  0.9998 &
  0.9282 &
  {  \textless{}0.0001$^{***}$} &
  0.9927 &
  0.0602$^{*}$ &
  0.0342$^{**}$ &
  0.9980 \\ \bottomrule 
\end{tabular}%
}
\end{table}

We further investigate the performance of the estimation methods beyond linearity. 
Table~\ref{tab:NonLinPtilde} summarize the power of the test at the Delta-band for the three settings, given $\bs{X}(t)$ and $\bs{Y}(t)$ have a non-linear relationship. 
Note that, for linear (relationship) case, the measure of effect size, $\Delta_1$, is straightforward to obtain since the eigenvalues and eigenvectors of true spectral density matrices can be obtained analytically (see Supplementary Material, Section A-3, for an example). However, the analytical expressions of eigenvalue and eigenvector in the non-linear case has no closed form. 
We then used two settings in the non-linear case (see Figure~\ref{SimSignals}) to allow for comparison of methods---(a) one setting has the same process for two states (i.e., zero change) and (b) another setting has different $\bs{E}^{(g)}$ for two states (i.e., non-zero change). 
From Table~\ref{tab:NonLinPtilde}, none of the estimation method rejects the null, when null is true. However, when $\Delta_1$ is non-zero, KenCoh consistently demonstrates the highest power, in fact, the other estimators did not detect the difference---i.e., only KenCoh detects significant differences across all cases.
These results indicate that KenCoh outperforms the other two methods, with its advantage becoming more pronounced when the relationships among variables are non-linear.


\section{Canonical Band-Coherence of Brain Signals}\label{chap:analysis}

We now investigate the functional brain connectivity structure of humans (in a driving experiment) and in mice (in an auditory experiment) using the novel KenCoh method and compare its performance against the baselines. In Section~\ref{subsec:analysis1}, we examine the levels of alertness and drowsiness through EEG recordings. In Section~\ref{subsec:analysis2}, we compared the neural activity captured through calcium imaging during in-sequence and oddball trials.

\subsection{Analysis of the Driving EEG Data}\label{subsec:analysis1}

The virtual driving experiment data in \cite{Cao2019DrivingData} consists of multiple trials per subject, with each trial lasting three seconds and the EEGs sampled at 128 Hz. Experts classified each trial as either an alert or drowsy state. 
We use the statistical tool proposed in this paper (KenCoh) to compare lagged cross-dependence between regions (i.e., $\bs{\Gamma}_{XY}(\Omega,\ell)$, for $\ell = 0, \pm1, \dots, \pm15$ in Equation~\eqref{cancor2}) between the alert ($g=1$) and drowsy ($g=0$) states, despite the heavy-tailed nature of the signals (see Figure~\ref{SamplePlot} showing Fp1 of an alert and a drowsy trial).  
\cite{Yantis2002FrontoPar} and \cite{liu2023BetaRightFronto} reveal that left fronto-parietal connectivity engaged in cognitive processes that require attention. Moreover, \cite{liu2023BetaRightFronto} mentioned that the right-frontal lobe is sensitive to distractions while driving and hence concentration is more evident in the left-hemisphere.
We, then, further subdivided the Fp and TPO regions (see Figure~\ref{framework}) into left and right hemispheres, resulting in four distinct groups of signals. These groups are as follows: (i) Left Fp region (LFp), including Fp1, F3, and F7; (ii) Right Fp region (RFp), including Fp2, F4, and F8; (iii) Left TPO region (LTPO), including T7, TP7, P3, and O1; and (iv) Right TPO region (RTPO), including T8, TP6, P4, and O2.

\begin{figure}[]
    \centering
    \includegraphics[width=0.75\textwidth]{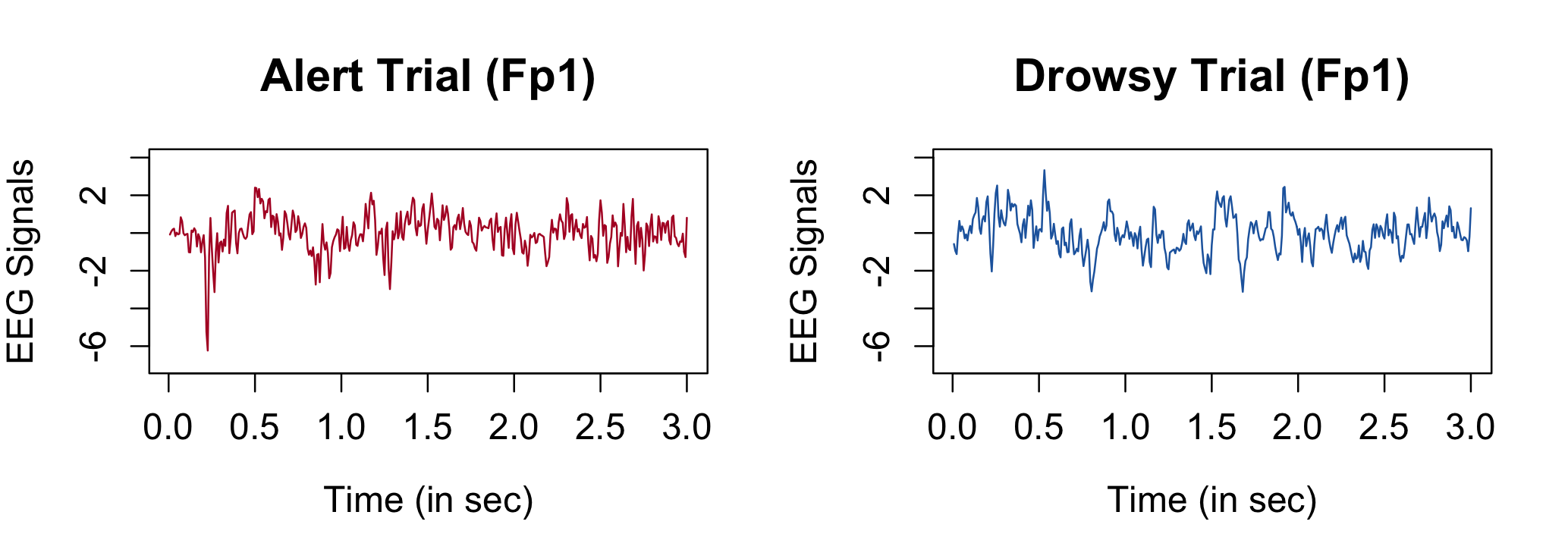}
    \caption{EEG signal at the Fp1 for one trial in alert state and drowsy state}
    \label{SamplePlot}
\end{figure}

We hypothesize that the brain network connectivity in alert and drowsy states are different and test the hypothesis using the method discussed in Section~\ref{subsec4:inference}.
We present here the results for the individual with the highest number of alert and drowsy trials. A complete analysis of all 11 individuals is provided in the Supplementary Material, Table A-4.2.
Table~\ref{tab:AdjPVals-notDir} summarizes the test results for the Delta frequency band (see Supplementary Material, Table A-4.1 for the results presented in greater details).
Adjustments for false discovery rate was done through \cite{benjamini2000adaptive} for multiple testing in different bands. Table~\ref{tab:AdjPVals-notDir} highlights that differences between alert and drowsy states are more pronounced in connectivity involving the left Fp lobe, particularly in the \textit{slow-wave band} (i.e., Delta-band). This significant difference was not detected using the other two methods. Recall that the simulation studies in Section~\ref{subsec:simresults}, KenCoh is more powerful when the data is heavy tailed or may have non-linear relationship among the variables. 

\begin{table}[]
\caption{The $p$-values at the Delta-band adjusted using the false discovery rate \citep{benjamini2000adaptive} with 5000 permutations. Significant $p$-values are marked with $^{**}$ if it is below 5\%.}
\label{tab:AdjPVals-notDir}
\resizebox{\textwidth}{!}{%
\begin{tabular}{@{}lcccccc@{}}
\toprule
\textbf{Connectivity}      & \textbf{LFp$\leftrightarrow$LTPO} & \textbf{LFp$\leftrightarrow$RTPO} & \textbf{LFp$\leftrightarrow$RFp} & \textbf{LTPO$\leftrightarrow$RTPO} & \textbf{LTPO$\leftrightarrow$RFp} & \textbf{RFp$\leftrightarrow$RTPO} \\ \midrule
{KenCoh $p$-values} & {0.0129$^{**}$}            & {0.0024$^{**}$}            & {0.0158$^{**}$}           & 0.7008               & 0.6713              & 0.4553              \\ 
VCov $p$-values            & 0.6132              & {0.0270$^{**}$}             & {0.0270$^{**}$}            & 0.4280                & 0.6132              & 0.4863              \\
MCD $p$-values             & 0.2184              & 0.8136              & 0.9200               & 0.9200                 & 0.9200                & 0.2184              \\ \bottomrule
\end{tabular}%
}
\end{table}
\begin{figure}
    \centering
    \includegraphics[width=0.85\textwidth]{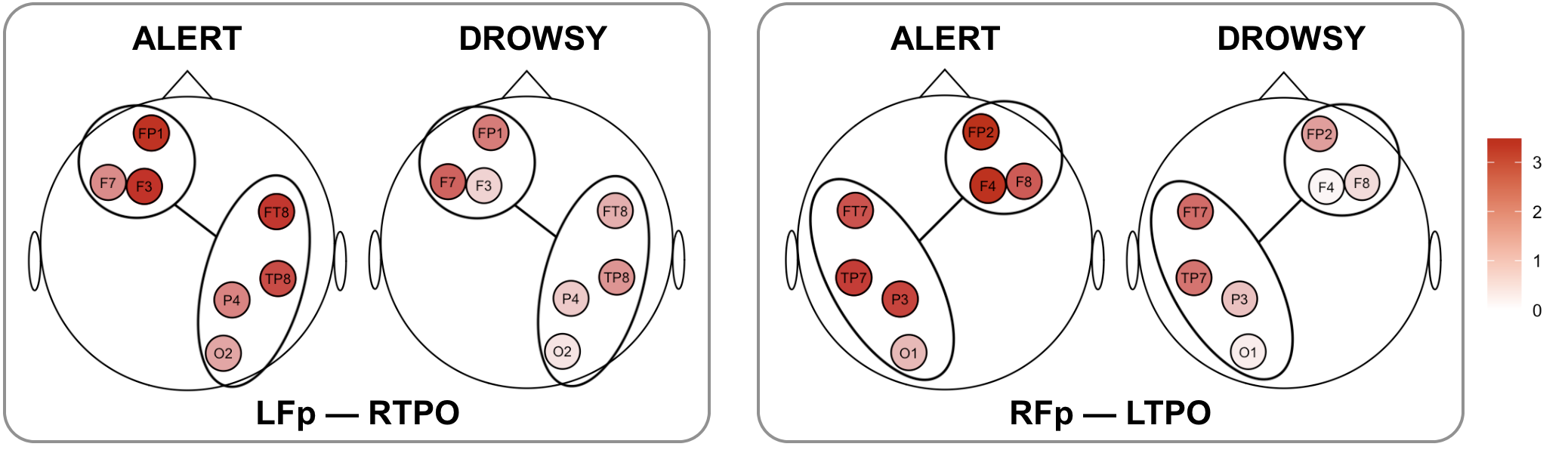}
    \caption{Mean estimates of absolute canonical directions via KenCoh in the Delta band for alert and drowsy states of subjects for LFp$\leftrightarrow$RTPO and RFp$\leftrightarrow$LTPO connectivities. 
    }
    \label{DeltaConnectivitySubj9}
\end{figure}
\begin{figure}
    \centering
    \includegraphics[width=0.85\textwidth]{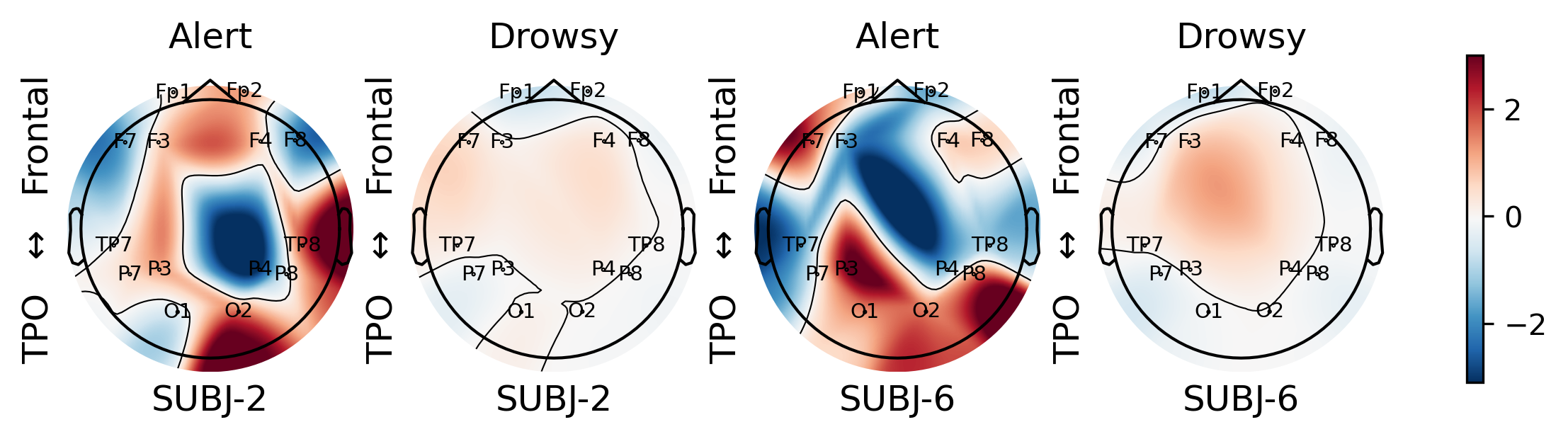}
    \caption{Mean estimates of canonical directions via KenCoh in the Delta band for alert and drowsy states of two other subjects with significant change in connectivity (of LFp$\leftrightarrow$RTPO and RFp$\leftrightarrow$LTPO). 
    }
    \label{DeltaConnectivityAllChannels}
\end{figure}

To further examine these differences at the Delta-band, Figure~\ref{DeltaConnectivitySubj9} shows the mean absolute canonical directions of 157 trials for alert and drowsy state using KenCoh method. Significant differences were found at LFp connectivities. Particularly, signals from the LFp region are associated differently with signals from the right TPO region (RTPO) in alert and drowsy states (see LFp $\leftrightarrow$ RTPO connectivity in Figure~\ref{DeltaConnectivitySubj9}). The right temporal-parietal channels (TP6 and T8)  \citep[both close to temporal lobe which is responsible for ``working memory for short term visual maintenance of information'', see][]{Patel2023Temporal}  and left frontal channel (F7) exhibit drastic changes in their weights when the state shifts. 
Specifically, all signals from the LFp region show drastic changes between the alert and drowsy states. In addition to this, Figure~\ref{DeltaConnectivityAllChannels} shows two other subjects with similar significant differences in connectivity between alert and drowsy states, specifically, the LFp $\leftrightarrow$ RTPO and RFp $\leftrightarrow$ LTPO connectivities at the Delta-band.
In summary, the frontal channel, which is the power-house in active thinking \citep{Yantis2002FrontoPar, liu2023BetaRightFronto}, has varying connectivity structure in the alert and drowsy states. The result indicates that association between frontal region (related to cognition) and parietal region (related to sense of perception) changes in the alert and drowsy state of an individual.

\subsection{Analysis of the Mouse Brain Calcium Signals} \label{subsec:analysis2}

\begin{figure}[]
    \centering
    \includegraphics[width= 0.9\textwidth]{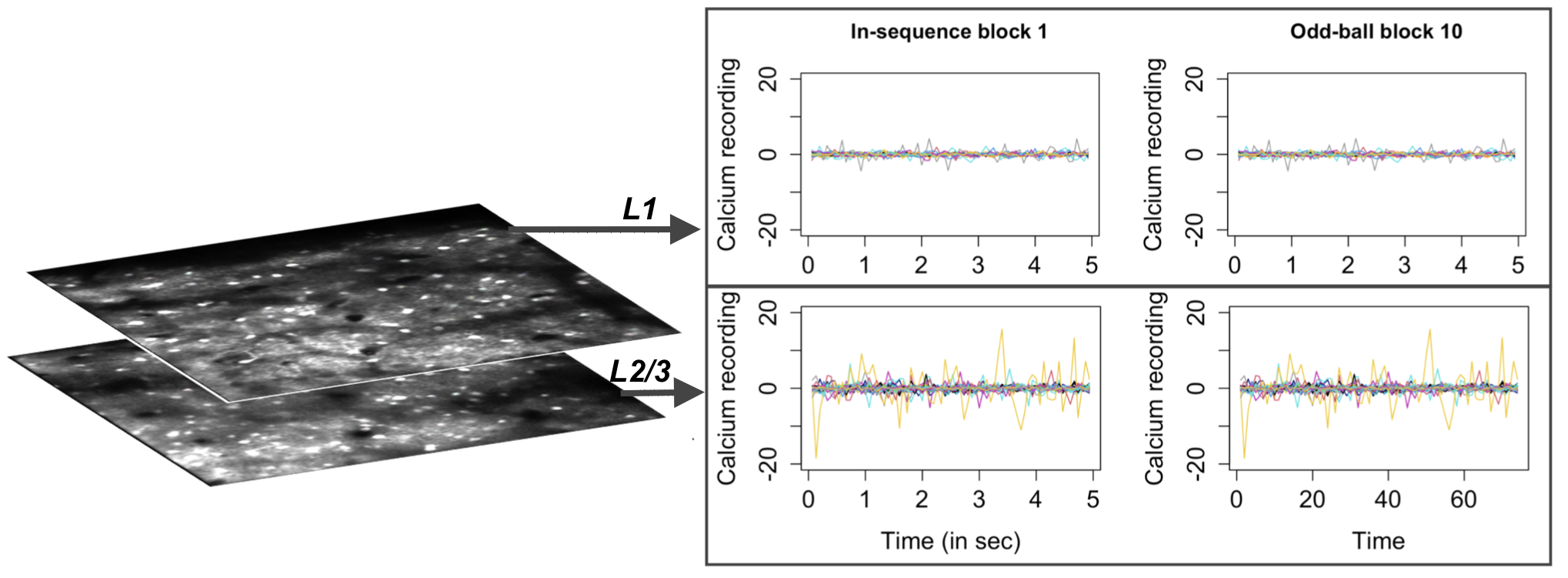}
    \caption{Calcium imaging data of the class of \textit{undefined-response neurons} in the first two in-sequence trial and one oddball sequence trial. Top panel shows L1 neurons and bottom panel shows the L2/3 neurons.}
    \label{SampleNeurons}
\end{figure} 
To showcase versatility of the method's application, we analyzed the mice calcium imaging data in \cite{reyesvallejo2023neurogliaform}. Calcium imaging enables monitoring of neuronal activity, in both inhibitory and excitatory neurons \citep{Grienberger2012imaging}. 
Our goal is to capture the cross-dependence between two neuronal populations: the first in layer 1 (L1) and the second in combined layers 2 and 3 (L2/3).
In our application, L1 is composed of 65 neurons and L2/3 is composed of 73 neurons.
The neurons were further classified into three sub-categories based on the homogeneity of neuronal response to repetitive stimuli presented. The three categories are (i) adapting response neurons (AN), or neurons that exhibit decreased activity in response to repetitive stimuli; (ii) increased-response neurons (IN), or neurons that show an increased response to repetitive stimuli; and (iii) undefined-response neurons (UN), or neurons with unclear or minimal change in activity in response to repetition. 
Table A-4.3 in the Supplementary Material shows the distribution of neurons based on its classification. 
Here, we provide a more in-depth analysis for the class of \underline{undefined-response neurons} (UN) to better understand how its activity changes in response to stimuli. Figure~\ref{SampleNeurons} shows sampled preprocessed signals from UN. The preprocessing includes centering the signals at zero and taking the first difference to eliminate the trend. 

Neuronal signals were sampled at 15 Hz and the duration of each trial was 5 seconds. Trials were categorized as either in-sequence or out-of-sequence---the latter also referred to as ``oddball" trials. 
The data for the first seven in-sequence trials (each replicated ten times) were analyzed,
resulting in a total of 70 in-sequence trials. In contrast, the oddball trials were replicated only nine times, and all were included in the analysis. 
A low-pass filter was applied to isolate low-frequency oscillations (i.e., below 4 Hz) and a high-pass filter was applied to obtain frequencies above 4 Hz, effectively separating low-frequency and high-frequency signals, respectively. 
The objective in this analysis is to estimate the lagged cross-dependence matrix at the low- and high-frequency signals through KenCoh method (see Section~\ref{subsec2:estimation}). 
\cite{Fu2014lowFreq} reveals that low-frequency oscillations in calcium imaging is associated with neuronal firings.
Hypothesis testing reveals that UNs---or neurons which have subtle change in neuronal firings---do not exhibit significant differences in slow-wave activity, at 5\% level of significance (see Table~\ref{tab:TestClassifNeurons}). However, their connectivity structure shows significant differences in high-frequency oscillations. In other words, our analysis revealed that in these UNs, the presentation of oddball in the experiment is related to the change in connectivity at higher-frequency oscillations.  
Although undefined-response neurons do not change their overall firing (slow, $<$4 Hz calcium fluctuations), oddball stimuli alter their inter-layer coordination seen as changes in the faster (4–7 Hz) oscillatory coupling, indicating a shift in timing/communication rather than rate. 

\begin{figure}
        \centering
        \includegraphics[width=1.0\textwidth]{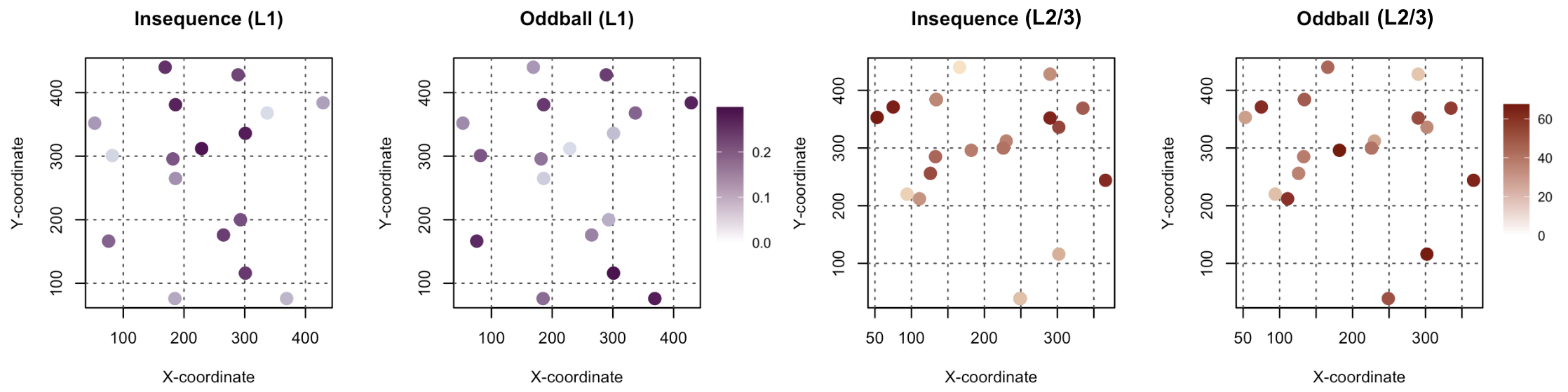}
        \caption{Average absolute canonical direction of UNs between L1 and L2/3, for oddball and {in-sequence} trials at \underline{high-frequency} band ({(4--7) Hz}).}
        \label{Layer1-23-HighFreq}
\end{figure} 

 \begin{table}
        \center
        \caption{Hypothesis testing results for detecting changes in connectivity structure between in-sequence trials and oddball trials for UNs. The $p$-values are adjusted using \cite{benjamini2000adaptive} for multiple testing. Significant $p$-values, after adjustment, are marked with $^{**}$ if it is below 5\%.}
        \label{tab:TestClassifNeurons}
        \begin{tabular}{@{}lrrr@{}}
        \toprule
        Frequency & VCov & KenCoh & MCD \\ \midrule
        Low-frequency $p$-value               & 0.3460 & 0.1502 & 0.2262   \\
        High-frequency $p$-value            & 0.6004 & 0.0016$^{**}$ & 0.2262   \\ \bottomrule
        \end{tabular}
        \end{table}
    
The KenCoh estimates of absolute mean canonical directions of UNs, filtered at high frequency, for in-sequence and oddball trials are shown in Figure~\ref{Layer1-23-HighFreq}. From this figure, observe that certain neurons in L1 and L2/3 are the primary drivers of association between these two layers during oddball, but not during in-sequence trials, and vice-versa.
KenCoh analysis reveals that communication between L1 and L2/3 is dependent on frequency; oddball stimuli drive certain neurons to strengthen their coordination in the faster (e.g., 4–7 Hz) band, even though their slow-wave activity stays the same. 
Theta-band inter-layer reconfiguration in UNs may serve as a latent signal of prediction error or novelty processing even when firing rates remain unchanged, offering a potential new biomarker for disorders of sensory integration. Future work should establish causality by perturbing top-down inputs to L1 and measuring the resulting effects on theta coupling, probe how neuromodulatory systems shape this coordination, and test whether the same signature is altered in disease models (e.g., ASD models such as Fmr1 KO). Translating this to humans could start with non-invasive measurements of theta coherence during oddball paradigms as early indicators of predictive-coding deficits, with the longer-term goal of using targeted neurostimulation or neurofeedback to restore or modulate maladaptive timing and coordination dynamics.

\section{Conclusion}\label{chap:conclusion}

Brain signals from various neurons or different channels exhibit natural groupings (as determined by the spatial configuration in both EEGs and calcium recordings).
Despite the importance of understanding the spectral dependence between such groups of time-series, there is a scarcity of studies in the literature addressing this topic of modeling dependence in the oscillatory activity between two groups of channels (or layers of neurons). One of the few available approaches is the canonical coherence analysis defined formally in \cite{brillinger2001time} \citep[which was inspired by the seminal work of canonical correlation analysis in][]{Hotelling1992CCA}. This method, however, is limited to singleton-frequency and has not been explored for heavy-tailed signals. One of the novel contributions of our paper is the definition of CBC (see Definition~\ref{definition:CBC}) which is applicable to frequency bands. Another novel contribution presented in this paper is the utilization of Kendall's tau for CBC estimation which leads to the advantage of robustness to outliers, and to a hypothesis testing framework that is sensitive to functional brain networks differences across brain states and experimental conditions.
Our proposed method, KenCoh, takes advantage of the spectral-dependence matrix,
and provides an interpretable and robust measure of channel weights (i.e., $\mbf{a}_\Omega$ and $\mbf{b}_\Omega$) that forms the connectivity structure between groups of signals. 

A simulation study was conducted to thoroughly examine the ability of the permutation test defined in Section~\ref{subsec4:inference} to detect changes in the connectivity structure via different estimation methods detailed in Section~\ref{subsec2:estimation}. These numerical experiments demonstrate that KenCoh is at par with the moment-based estimator of the covariance matrix under the assumption of Gaussianity. 
Another distinct advantage of KenCoh, over the classical measures, is its ability to capture non-linear dependence, even in the presence of outliers.
Moreover, if the underlying distributions are heavy-tailed, KenCoh outperforms the moment-based estimator and MCD, a robust estimator of covariance. 

The analysis of EEG data using the proposed KenCoh method uncovered novel findings in the analysis of brain signals. One of these is the stronger contribution of LFp- and parietal-channels during alert state. This provides another empirical evidence that supports the neuroscience literature where this unique connection is present when an individual is doing a physical activity \citep[]{liu2023BetaRightFronto,Yantis2002FrontoPar}.
Moreover, we detected unique neuronal behavior in calcium imaging of mice when studying the association between L1 and L2/3 neurons. Specifically, neurons have distinct characteristics for different frequency bands, such as the connectivity observed for class of UNs---i.e., some neurons have stronger contributions to $\kappa_\Omega$ during oddball trial when higher frequency oscillations are isolated. 

There are some limitations in the current approach that will be addressed in future studies. Here, trials are assumed to be uncorrelated (though the neurons are dependent within each trial). Our future work will involve modeling some temporal dependence across trials. This is not straightforward as this would require developing a framework that extends the evolving evolutionary processes by \cite{fiecas2016modeling}. In addition, for multi-subject data, our future work will be to rigorously develop some ``mixed effects framework'' following the mixed-effects vector autoregressive model (ME-VAR) in \cite{gorrostieta2012investigating}. In this set-up, we will define some ``fixed effects'' to denote the specific functional brain networks for each brain state and for each experimental conditions. We will also define ``random effects'' to capture variation between trials (within a subject) and the variation across participants in the study. These are all important directions that we anticipate to be impactful in the study of physiological functional brain processes.

\begin{appendix}
\section{Supplementary Material} \label{appn1} 
This supplementary document provides additional details to support the methodology and results presented in the main manuscript. 

\section{Proof of Propositions} \label{appn1} 

\subsection{Proof of Proposition 3.1}
\begin{duplicate}
    Let $c_{\Omega}(s), \ s = 0, \pm 1, \pm 2, \ldots$, be a linear filter, such that, ${\sum_{s = -\infty}^{\infty} |c_{\Omega}(s)| < \infty}$. Moreover, let ${\{\mbf{Z}(t)\}_{t = 1}^T}$ be a weakly stationary $D$-dimensional time-series as defined in Definition 3.1. Let ${\mbf Z}^\Omega(t)$ be the convolution of $\mbf{Z}(t)$ and the filter $c_{\Omega}(h)$, for $t = 1, \dots, T$. We define $\bs{\Gamma}(\Omega, \ell)$ as the lagged dependence matrix of ${\mbf Z}^\Omega(t)$ as defined in Definition 3.1, for ${\ell = 0, \pm 1, \dots, \pm L}$. If second-order moments of $\{\bs{Z}(t)\}_{t = 1}^T$ exists, then canonical band-coherence, $\kappa(\Omega)$, is equivalent to,
\begin{gather*}
    \kappa(\Omega) = \max_{ \{\mbf{a}_\Omega , \mbf{b}_\Omega, \ell \}} \left| \mbf{a}^{\top}_\Omega \left[ \int_{\{\omega : \Omega/S\}} \mbf{f}_{XY}(\omega) e^{\rmi 2\pi\omega \ell} \rmd \omega \right] \mbf{b}_\Omega \right|^2 \;\;  \\
    \text{such that, } {\mbf a}_\Omega^\top {\mbf f}_{XX}(\Omega) {\mbf a}_\Omega = {\mbf b}_\Omega^\top {\mbf f}_{YY}(\Omega)  {\mbf b}_\Omega = 1, 
\end{gather*}
where $\mbf{a}_{\Omega} \in \mathbb{R}^P$ and $\mbf{b}_{\Omega} \in \mathbb{R}^Q$ are the canonical directions.
\end{duplicate}
\begin{proof}
    Let $\bs{\Sigma}_{XY,\Omega}(\ell)$ be the covariance matrix of filtered series $\mbf{X}^{\Omega}(t)$ and $\mbf{Y}^{\Omega}(t)$ at band-$\Omega$ and lag-$\ell$, for $\ell = 0, 1, \dots, T-1$. Moreover, let $\bs{\Sigma}_{XX,\Omega}$ and $\bs{\Sigma}_{YY,\Omega}$ be the covariance matrix of filtered series $\mbf{X}^{\Omega}(t)$ and ${\mbf Y}^{\Omega}(t)$ at $\ell = 0$, respectively. We have ${\bs{\Gamma}_{XY}(\Omega, \ell) = K \bs{\Sigma}_{XY,\Omega}(\ell)}$, where $K = -2\psi'(0)$ is constant. Correspondingly, we define ${\bs{\Gamma}_{XX}(\Omega, \ell) = K \bs{\Sigma}_{XX,\Omega}(\ell)}$ and $\bs{\Gamma}_{YY}(\Omega, \ell) = K \bs{\Sigma}_{YY,\Omega}(\ell)$. 
    Given ${{\mbf a}^\top_\Omega \bs{\Gamma}_{XX}(\Omega, 0) {\mbf a}_\Omega = {\mbf b}_\Omega^\top \bs{\Gamma}_{YY}(\Omega, 0) {\mbf b}_\Omega = 1}$, then, 
    \begin{eqnarray*}
        \kappa_{XY}^{1/2} (\Omega) & = & \max_{  \{{\mbf a}_{\Omega}, {\mbf b}_{\Omega}; \ell \}} {\mbf a}^{\top}_\Omega  \bs{\Gamma}_{XY}(\Omega,\ell) {\mbf b}_{\Omega} \\
        & = & \max_{\{{\mbf a}_\Omega, {\mbf b}_\Omega; \ell \}} \frac{ {\mbf a}_\Omega^\top \bs{\Gamma}_{XY}(\Omega, \ell) {\mbf b}_\Omega }{ \sqrt{{\mbf a}^\top_\Omega \bs{\Gamma}_{XX}(\Omega, 0) {\mbf a}_\Omega {\mbf b}^\top_\Omega \bs{\Gamma}_{YY}(\Omega, 0) {\mbf b}_\Omega} } \\
        & = & \max_{\{{\mbf a}_\Omega, {\mbf b}_\Omega; \ell \}} \frac{ {\mbf a}_\Omega^\top \bs{\Sigma}_{XY,\Omega}(\ell) {\mbf b}_\Omega }{ \sqrt{{\mbf a}^\top_\Omega \bs{\Sigma}_{XX,\Omega}(0) {\mbf a}_\Omega {\mbf b}^\top_\Omega \bs{\Sigma}_{YY,\Omega}(0) {\mbf b}_\Omega} }
    \end{eqnarray*}
    
    Given a filtered series  $X^{\Omega}(t)  = \sum_{s = -\infty}^\infty c_{\Omega}(s) X(t-s)$, where $c_{\Omega}(h^*)$ is a linear filter that has zero spectrum outside a defined band, e.g., 
    $$|C^{(\Omega)}(\omega)|^2 = |\sum_{ s} c_{\Omega}(s) e^{-\rmi 2\pi\omega s}|^2 = 1/2\delta,$$ 
    for $\omega \in \Omega/S$ and 0 otherwise \citep{Ombao2008evolutionary}. Then $f_{X}(\Omega) =  |C^{(\Omega)}(\omega)|^2 f_{X}(\omega)$ \cite{lindgren2013stationary} and $\text{Cov}(X(t-\ell),X(t)) = \int_{-0.5}^{0.5} f_{X}(\omega) e^{\rmi 2\pi\omega \ell} \rmd \omega$. Similarly, it can be shown that $f_{XY}(\Omega) =  |C^{(\Omega)}(\omega)|^2 f_{XY}(\omega)$ and $\text{Cov}(X(t-\ell),Y(t)) = \int_{-0.5}^{0.5} f_{XY}(\omega) e^{\rmi 2\pi\omega \ell} \rmd \omega$. Hence,
    \begin{eqnarray*}
        \kappa_{XY}^{1/2} (\Omega) & = & \max_{  \{{\mbf a}, {\mbf b}; \ell \}} \frac{ {\mbf a}^\top \int_{-0.5}^{0.5}|C^{(\Omega)}(\omega)|^2 {\mbf f}_{XY}(\omega) e^{\rmi 2\pi\omega \ell} \rmd\omega {\mbf b} }{ \sqrt{{\mbf a}^\top \int_{-0.5}^{0.5}|C^{(\Omega)}(\omega)|^2 {\mbf f}_{XX}(\omega) \rmd\omega {\mbf a} {\mbf b}^\top \int_{-0.5}^{0.5}|C^{(\Omega)}(\omega)|^2 {\mbf f}_{YY}(\omega) \rmd\omega {\mbf b}} } \\
        & = & \max_{  \{{\mbf a}, {\mbf b}; \ell \}} \frac{ {\mbf a}^\top \int_{\Omega} {\mbf f}_{XY}(\omega) e^{\rmi 2\pi\omega \ell} \rmd\omega {\mbf b} }{ \sqrt{{\mbf a}^\top {\mbf f}_{XX}(\Omega) {\mbf a} {\mbf b}^\top {\mbf f}_{YY}(\Omega)  {\mbf b}} }
    \end{eqnarray*}
Given ${\mbf a}^\top {\mbf f}_{XX}(\Omega) {\mbf a} = {\mbf b}^\top {\mbf f}_{YY}(\Omega)  {\mbf b} = 1$, we have 
   \[ \kappa(\Omega) = \max_{ \{\mbf{a}_\Omega , \mbf{b}_\Omega, \ell \}} \left| \mbf{a}^{\top}_\Omega \left[ \int_{\{\omega : \Omega/S\}} \mbf{f}_{XY}(\omega) e^{\rmi 2\pi\omega \ell} \rmd \omega \right] \mbf{b}_\Omega \right|^2. \]
    
\end{proof}

\subsection{Proof of Proposition 3.2}

\begin{duplicate}
   Consider the weakly stationary time series filtered at frequency band $\Omega$, ${(\bs{X}^{\Omega\top}(t), \bs{Y}^{\Omega\top}(t))^\top \in \mathbb{R}^{P+Q}}$, and its lagged cross-dependence matrix $\bs{\Gamma}(\Omega, \ell)$ as in Definition 3.1, for $\ell = 0, \pm 1, \dots, \pm L$.
Let ${\Lambda}_j(\Omega, \ell)$ be the $j$-th largest eigenvalue of the following:
    \begin{gather}
    \bs{\Gamma}_{XX}^{-1/2}(\Omega, 0)\bs{\Gamma}_{XY}(\Omega, \ell)\bs{\Gamma}_{YY}^{-1}(\Omega, 0)\bs{\Gamma}_{YX}(\Omega, -\ell)\bs{\Gamma}_{XX}^{-1/2}(\Omega, 0), \; \;  \label{Pxy} \\
    \bs{\Gamma}_{YY}^{-1/2}(\Omega, 0)\bs{\Gamma}_{YX}(\Omega, -\ell)\bs{\Gamma}_{XX}^{-1}(\Omega, 0)\bs{\Gamma}_{XY}(\Omega, \ell)\bs{\Gamma}_{YY}^{-1/2}(\Omega, 0). \label{Pyx} 
\end{gather}
for $j = 1, \dots, \min(P,Q)$. Moreover, for non-zero $\Lambda_{1}(\Omega, \ell)$, we denote the corresponding eigenvectors for Equations~\eqref{Pxy} and \eqref{Pyx} at lag $\ell$, as $\mbf{u}_{\Omega} \in \mathbb{R}^P$ and $\mbf{v}_{\Omega} \in \mathbb{R}^Q$, respectively.
Then, $\kappa_\Omega = \max_{\ell} \{{\Lambda}_{1}(\Omega, \ell)\}$,
    \begin{equation}
       \mbf{a}_{\Omega} := \bs{\Gamma}_{XX}(\Omega, 0)^{-1/2}\mbf{u}_\Omega \ \text{and } \  \mbf{b}_{\Omega} := \bs{\Gamma}_{YY}(\Omega, 0)^{-1/2}\mbf{v}_\Omega. \label{canvec}
    \end{equation}
\end{duplicate}

\begin{proof}
    The proof follows the argument similar to Theorem 11.2.1 in \cite{Mardia1979bibby}. Let ${\bs{\Theta}}(\ell) = \{\bs{{\Gamma}}_{{XX}}(\Omega,0)\}^{-\frac{1}{2}}\bs{{\Gamma}}_{{XY}}(\Omega,\ell)\{\bs{{\Gamma}}_{{YY}}(\Omega,0)\}^{-\frac{1}{2}}$. 
    From here, we drop the $\Omega$ for brevity.
    We then have
    \[\max\limits_{\bs u, v, \ell}\left [\bs{u}^\top\bs{\Theta}(\ell)\bs{v}\right ]^2 = \max\limits_{\bs{u}, \bs{v}, \ell}\ \ \bs{u}^\top {\bs{\Theta}(\ell)\bs{v}\bs{v}^\top\bs{\Theta}^\top(\ell)}\bs{u}.\]
    For fixed $\bs{v}$, we define the following decomposition: $${\bs{\Theta}(\ell)\bs{v}\bs{v}^\top\bs{\Theta}^\top(\ell)} = \bs{D}\bs{\Lambda}^{(\bs{v})}\bs{D}^\top.$$
     Let $\bs{\gamma}^{(\bs{v})} = \bs{D}^\top\bs{u}$, where $\bs{\gamma}^{\top(\bs{v})}\bs{\gamma}^{(\bs{v})} = 1$. Then,
        \[\max\limits_{\bs u, v}\left [\bs{u}^\top\bs{\Theta}(\ell)\bs{v} \right ]^2 =  \bs{\gamma}^{\top(\bs{v})}\bs{\Lambda}^{(\bs{v})}\bs{\gamma}^{(\bs{v})} \equiv \sum_{j=1}^P\Lambda^{(\bs{v})}_{j}\{\gamma^{(\bs{v})}_{j}\}^2.\]
        
    Note that ${\bs{\Theta}(\ell)\bs{v}\bs{v}^\top\bs{\Theta}^\top(\ell)}$ has same non-zero eigen-values with that of $\bs{v}^\top\bs{\Theta}^\top(\ell)\bs{\Theta}(\ell)\bs{v}$ which itself is scalar, and hence is its own eigenvalue \citep{Mardia1979bibby}. 
        Moreover, we have $\sum_j\Lambda_j^{(\bs{v})}\{\gamma^{(\bs{v})}_{j}\}^2 \leq \Lambda_1^{(\bs{v})}\sum_j \{\gamma^{(\bs{v})}_{j}\}^2$.
       In other words, $\Lambda_1^{(\bs{v})}=\bs{v}^\top\bs{\Theta}^\top(\ell)\bs{\Theta}(\ell)\bs{v}$.
        Hence the solution to $\bm{u}_\Omega$ is the leading eigenvector of ${\bs{\Theta}(\ell)\bs{v}\bs{v}^\top\bs{\Theta}^\top(\ell)}$.

        Using the same principle, we obtain the solution for $\bm{v}_\Omega$ through maximizing ${\bs{v}^\top\bs{\Theta}(\ell)^\top\bs{\Theta}(\ell)\bs{v}}$
       In summary, ${\kappa}(\Omega)$ is the largest eigenvalue of 
        \begin{gather*}
            \bs{\Theta}^\top(\ell)\bs{\Theta}(\ell) = \bs{\Gamma}_{YY}^{-1/2}(\Omega, 0)\bs{\Gamma}_{YX}(\Omega, -\ell)\bs{\Gamma}_{XX}^{-1}(\Omega, 0)\bs{\Gamma}_{XY}(\Omega, \ell)\bs{\Gamma}_{YY}^{-1/2}(\Omega, 0) \; \text{and} \\ 
            \bs{\Theta}(\ell)\bs{\Theta}^\top(\ell) = \bs{\Gamma}_{XX}^{-1/2}(\Omega, 0)\bs{\Gamma}_{XY}(\Omega, \ell)\bs{\Gamma}_{YY}^{-1}(\Omega, 0)\bs{\Gamma}_{YX}(\Omega, -\ell)\bs{\Gamma}_{XX}^{-1/2}(\Omega, 0).
        \end{gather*}

        \noindent for all $\ell = 0, \pm1, \dots, \pm L$. Additionally, by definition, we have $\kappa_\Omega = \max_{\ell} \{\bs{a}^\top_{\Omega}\bs{\Gamma}_{XY}(\Omega, \ell)\bs{b}_{\Omega}\}$. Hence, in 
        $$\max\limits_{\bs u_\Omega, v_\Omega, \ell}\left [\bs{u}_\Omega^\top\bs{\Theta}(\ell)\bs{v}_\Omega\right ]^2 = \max\limits_{\bs u_\Omega, v_\Omega, \ell}\left [\bs{u}_\Omega^\top\{\bs{{\Gamma}}_{{XX}}(\Omega,0)\}^{-\frac{1}{2}}\bs{{\Gamma}}_{{XY}}(\Omega,\ell)\{\bs{{\Gamma}}_{{YY}}(\Omega,0)\}^{-\frac{1}{2}}\bs{v}_\Omega\right ]^2,$$
        it follows that $\mbf{a}_{\Omega} := \bs{\Gamma}_{XX}(\Omega, 0)^{-1/2}\mbf{u}_\Omega \ \text{and } \  \mbf{b}_{\Omega} := \bs{\Gamma}_{YY}(\Omega, 0)^{-1/2}\mbf{v}_\Omega.$
\end{proof}

\section{Minimum Covariance Determinant} \label{appn3}

Following the discussion in \citep{hubert2018minimum}, suppose we find a subset of time series denoted as $\bs{Z}^{\Omega}_0(t,\ell) = \{(X_1(t), \dots, X_P(t), Y_1(t+\ell), \dots, Y_Q(t+\ell))\}_{t \in \mathcal{H}}$ such that $\mathcal{H}$ is the set of time points of length $T/2 < |\mathcal{H}| < T$ where determinant of $\text{Cov}(\bs{Z}_0^{\Omega}(t,\ell))$ is minimized \citep[see][ for details]{hubert2018minimum}. Let $\bs{\mu}_{0}(\ell) = \sum_{t \in \mathcal{H}} \bs{Z}^{\Omega}_0(t,\ell) / |\mathcal{H}|$ and $\bs{\Sigma}_{0}(\ell) = \text{Cov}(\bs{Z}^{\Omega}_0(t,\ell))$. Let $$R_t(\ell) = \sqrt{\bs{Z}^{\Omega}_0(t,\ell) - \bs{\mu}_{0}(\ell))^\top (\bs{\Sigma}_{0}(\ell))^{-1}(\bs{Z}^{\Omega}_0(t,\ell) - \bs{\mu}_{0}(\ell))}$$ and 
$w(R^2_t(\ell)) = \mathbb{I}\{R_t^2(\ell) \leq \chi^2_{D,0.975}\}$. The MCD estimator is given as
\begin{align}
    \hat{\mu}_{j,\text{MCD}}(\ell) &= \frac{\sum_{t = 1}^T w(R^2_t(\ell)) {Z^{\Omega}_j}(t,\ell)}{\sum_{t = 1}^T w(R^2_t(\ell))}, \text{ for } j = 1, \dots, D, \notag \\
    \hat{\bs{\Sigma}}^{\text{MCD}}(\Omega, \ell) &= K_1\frac{1}{T}\sum_{  t} w(R^2_t(\ell)) (\bs{Z}^{\Omega}_0(t,\ell) - \hat{\bs{\mu}}_{\text{MCD}}(\ell))^\top (\bs{Z}^{\Omega}_0(t,\ell) - \hat{\bs{\mu}}_{\text{MCD}}(\ell)), \notag \\
    \hat{\mathbb{M}}(\Omega, \ell) &= \hat{\bs{\Sigma}}_{\text{MCD}}^{-1/2}(\Omega,0) \hat{\bs{\Sigma}}_{\text{MCD}}(\Omega,\ell) \hat{\bs{\Sigma}}_{\text{MCD}}^{-1/2}(\Omega,0),  \label{MCD}
\end{align} 
where $K_1$ is the consistency factor defined in \citep{hubert2018minimum}.
The estimator $\hat{\mathbb{M}}(\Omega, \ell)$ is what we use to serve as a robust estimator for $\bs{\Gamma}(\Omega, \ell)$.

\section{Canonical directions for the AR(2) mixture model}

Let ${ f}^{(m)}_O(\omega)$ be the spectral dentsity of $O_m(t)$, for $m = 1, \dots, 5$. 
For brain state $g$, the auto-spectra of the $j$-th EEG and the cross-spectra between the $j$-th and $k$-th EEGs are, respectively,
\begin{gather}
    {f}_{jj}^{(g)}(\omega) = \{E_{j,1}^{(g)}\}^2 \ { f}^{(1)}_O(\omega) + \dots + \{E_{j,5}^{(g)}\}^2 \ {f}^{(5)}_O(\omega) + \sigma^2_W, \label{fjj} \\
    {f}_{jk}^{(g)}(\omega) =  E_{j,1}^{(g)}E_{k,1}^{(g)} { f}^{(1)}_O(\omega) + \dots + E_{j,5}^{(g)}E_{k,5}^{(g)}{ f}^{(5)}_O(\omega). \label{fjk}
\end{gather}

Without loss of generality, we provide in the equation below the analytical solution to $\kappa(\Omega)$ when $P = Q = 2$. Let $z_{m} = {f_O^{(m)}(\omega)}$, 
$U^{(g)}_1 = ({ \{E^{(g)}_{1,m}\}^2 \{E^{(g)}_{3,m}\}^2  - \{E^{(g)}_{1,m}\}^2 \{E^{(g)}_{4,m}\}^2 })$, 
${U^{(g)}_2 = ({ \{E^{(g)}_{2,m}\}^2 \{E^{(g)}_{3,m}\}^2  - \{E^{(g)}_{2,m}\}^2 \{E^{(g)}_{4,m}\}^2 })}$, and 
${U^{(g)}_3 = ({ E^{(g)}_{1,m}E^{(g)}_{2,m}\{E^{(g)}_{3,m}\}^2 - E^{(g)}_{1,m}E^{(g)}_{2,m}\{E^{(g)}_{4,m}\}^2})}$, we obtain the analytical solution by solving for $\Lambda$ in,
\begin{equation}
 k(E^{(g)}_{1,m},E^{(g)}_{2,m}) \ k(E^{(g)}_{3,m},E^{(g)}_{4,m}) \left[ (U^{(g)}_1  - \Lambda)(U^{(g)}_2 - \Lambda) - \{U^{(g)}_3\}^2 \right] = 0,  \label{pq2}
\end{equation}
where $k(x,y) = \frac{z_{m}\sigma_W^2}{(\{x\}^2z_{m} + \sigma^2_W)(\{y\}^2z_{m} + \sigma^2_W) - ({xy}z_{m})^2}$. 

\section{Additional Tables and Figures} \label{appn2}

This section contains some additional tables and figures that support the notions presented in the main text. 

\subsection{Analysis of Driving Data}
Here, we provide additional figures for the analysis of driving-data in \cite{Cao2019DrivingData}. 
In Figure~\ref{CanCohEst-wPCA} we show the overall coherence measure from four different methods for all 157 alert trials of one subject.

The four aggregation methods applied are---(1) averaging signals with equal weights (referred to as the Naive Estimator), (2) averaging signals using the first principal component from PCA (referred to as the PCA method) (3) canonical band-coherence using KenCoh and (4) canonical band-coherence using VCov. Figure~\ref{CanCohEst-wPCA} shows the measure of global coherence using these different methods.
\begin{figure}[h!]
    \centering
    \includegraphics[width=0.99\textwidth]{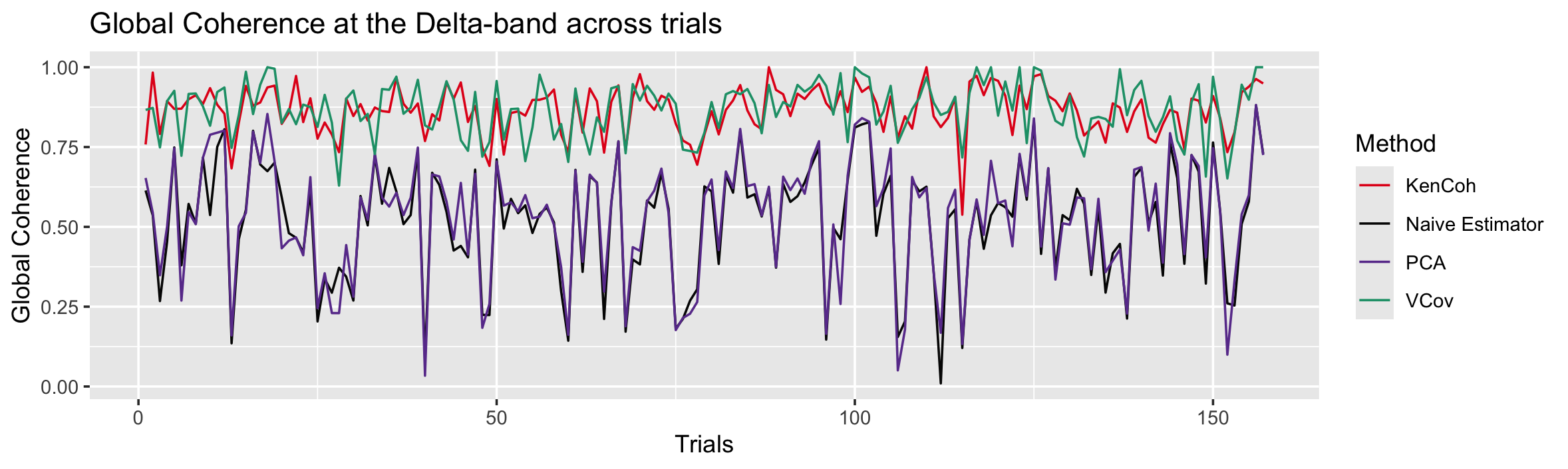}
    \caption{Comparison of global estimate for coherence of LFp (Left Frontal-Anterior region) and LTPO (Left Temporal-Parietal-Occipital region) using KenCoh, Naive Estimator, PCA, and VCov across the 157 trials of alert-state}
    \label{CanCohEst-wPCA}
\end{figure}
\begin{table}[h!]
\caption{Adjusted $p$-values using three estimation methods for all frequency-bands. Significant $p$-values are marked with $^{*}$ if it is below 10\%, $^{**}$ if it is below 5\%, and $^{***}$ if it is below 1\%.}
\label{tab:AdjPVals-notDir-allbands}
\resizebox{\textwidth}{!}{%
\begin{tabular}{@{}cccccccc@{}}
\toprule
\multicolumn{1}{l}{\textbf{Connectivity}} & \textbf{Connectivity} & \multicolumn{1}{c}{\textbf{LFp - LTPO}} & \multicolumn{1}{c}{\textbf{LFp - RTPO}} & \multicolumn{1}{c}{\textbf{LFp - RFp}} & \multicolumn{1}{c}{\textbf{LTPO - RTPO}} & \multicolumn{1}{c}{\textbf{LTPO - RFp}} & \multicolumn{1}{c}{\textbf{RFp - RTPO}} \\ \midrule
\multirow{5}{*}{\textbf{KenCoh}}          & Delta                 & \multicolumn{1}{c}{0.0129$^{**}$}            & \multicolumn{1}{c}{0.0024$^{**}$}            & \multicolumn{1}{c}{0.0158$^{**}$}           & \multicolumn{1}{c}{0.7008}               & \multicolumn{1}{c}{0.6713}              & \multicolumn{1}{c}{0.4553}              \\ 
                                          & Theta                 & \multicolumn{1}{c}{0.6092}              & \multicolumn{1}{c}{0.5872}              & \multicolumn{1}{c}{0.6226}             & \multicolumn{1}{c}{0.6152}               & \multicolumn{1}{c}{0.5996}              & \multicolumn{1}{c}{0.6074}              \\
                                          & Alpha                 & \multicolumn{1}{c}{0.2966}              & \multicolumn{1}{c}{0.3098}              & \multicolumn{1}{c}{0.303}              & \multicolumn{1}{c}{0.3082}               & \multicolumn{1}{c}{0.3074}              & \multicolumn{1}{c}{0.2978}              \\
                                          & Beta                  & \multicolumn{1}{c}{0.035$^{**}$}              & \multicolumn{1}{c}{0.036$^{**}$}              & \multicolumn{1}{c}{0.0328$^{**}$}            & \multicolumn{1}{c}{0.036$^{**}$}               & \multicolumn{1}{c}{0.0306$^{**}$}             & \multicolumn{1}{c}{0.0364$^{**}$}             \\
                                          & Gamma                 & \multicolumn{1}{c}{0.0006$^{**}$}             & \multicolumn{1}{c}{0.0012$^{**}$}             & \multicolumn{1}{c}{0.0006$^{**}$}            & \multicolumn{1}{c}{0.0004$^{**}$}              & \multicolumn{1}{c}{0.0008$^{**}$}             & \multicolumn{1}{c}{0.0014$^{**}$}             \\ \midrule
\multirow{5}{*}{\textbf{VCov}}            & Delta                 & 0.6132                                  & 0.027$^{**}$                                 & 0.027$^{**}$                                & 0.428                                    & 0.6132                                  & 0.4863                                  \\
                                          & Theta                 & 0.2104                                  & 0.2122                                  & 0.2252                                 & 0.2194                                   & 0.217                                   & 0.2142                                  \\
                                          & Alpha                 & 0.0054$^{**}$                                 & 0.0062$^{**}$                                 & 0.0044$^{**}$                                & 0.0042$^{**}$                                  & 0.0046$^{**}$                                 & 0.0034$^{**}$                                 \\
                                          & Beta                  & 0.0286$^{**}$                                 & 0.032$^{**}$                                  & 0.0318$^{**}$                                & 0.0312$^{**}$                                  & 0.0296$^{**}$                                 & 0.028$^{**}$                                  \\
                                          & Gamma                 & 0.0228$^{**}$                                 & 0.0274$^{**}$                                 & 0.0238$^{**}$                                & 0.0242$^{**}$                                  & 0.0264$^{**}$                                 & 0.0262$^{**}$                                 \\ \midrule
\multirow{5}{*}{\textbf{MCD}}             & Delta                 & 0.2184                                  & 0.8136                                  & 0.9200                                   & 0.92                                     & 0.9200                                    & 0.2184                                  \\
                                          & Theta                 & 0.3958                                  & 0.4134                                  & 0.4064                                 & 0.4006                                   & 0.4112                                  & 0.3938                                  \\
                                          & Alpha                 & 0.9596                                  & 0.9604                                  & 0.9512                                 & 0.9590                                    & 0.9522                                  & 0.9604                                  \\
                                          & Beta                  & 0.0952$^{*}$                                  & 0.0988$^{*}$                                  & 0.1014                                 & 0.1066                                   & 0.1056                                  & 0.0956$^{*}$                                  \\
                                          & Gamma                 & \multicolumn{1}{c}{0.0228$^{**}$}             & \multicolumn{1}{c}{0.0274$^{**}$}             & \multicolumn{1}{c}{0.0238$^{**}$}            & \multicolumn{1}{c}{0.0242$^{**}$}              & \multicolumn{1}{c}{0.0264$^{**}$}             & \multicolumn{1}{c}{0.0262$^{**}$}             \\ \bottomrule 
\end{tabular}%
}
\end{table}

\begin{table}[]
\caption{The $p$-values at the Delta-band for all subject. Significant $p$-values are marked with $^{*}$ if it is below 10\%, $^{**}$ if it is below 5\%, and $^{***}$ if it is below 1\%.}
\label{tab:AllSubj-Delta}
\resizebox{\textwidth}{!}{%
\begin{tabular}{@{}lcccccc@{}}
\toprule
Connectivity & LFp$\leftrightarrow$LTPO & LFp$\leftrightarrow$RTPO & LFp$\leftrightarrow$RFp & LTPO$\leftrightarrow$RTPO & LTPO$\leftrightarrow$RFp & RFp$\leftrightarrow$RTPO \\ \midrule
Subj 1       & 0.5284     & 0.0227$^{**}$     & $<0.0001^{***}$         & 0.0045$^{**}$      & 0.1086     & 0.0160$^{**}$      \\
Subj 2       & $<0.0001^{***}$          & 0.0153$^{**}$     & 0.2742    & $<0.0001^{***}$           & $<0.0001^{***}$          & 0.3678     \\
Subj 3       & 0.3210      & 0.3110      & 0.9948    & 0.1760       & 0.8648     & 0.115      \\ 
Subj 4       & 0.8078     & 0.0043$^{**}$     & 0.0283$^{**}$    & 0.0947$^{*}$      & $<0.0001^{***}$          & 0.0640$^{*}$      \\
Subj 5       & 0.0752$^{*}$     & 0.2980      & 0.0005$^{**}$     & 0.0616$^{**}$      & $<0.0001^{***}$          & 0.161      \\
Subj 6       & $<0.0001^{***}$          & $<0.0001^{***}$          & 0.0003$^{**}$     & 0.0063$^{**}$      & $<0.0001^{***}$          & $<0.0001^{***}$          \\
Subj 7       & 0.1680      & 0.0235$^{**}$     & 0.042$^{**}$     & 0.9404      & 0.6490      & 0.6700       \\
Subj 8       & 0.5678     & 0.0397$^{**}$     & 0.1077    & 0.4236      & 0.7600       & 0.2034     \\
Subj 9       & 0.0215$^{**}$     & 0.022$^{**}$      & 0.0254$^{**}$    & 0.3960       & 0.2770      & 0.2703     \\
Subj 10      & 0.5365     & 0.4470      & 0.2390     & 0.6110       & 0.3820      & 0.7908     \\ 
Subj 11      & 0.5094     & 0.7254     & 0.4542    & 0.5430       & 0.4800       & 0.1968     \\ \bottomrule
\end{tabular}%
}
\end{table}

\subsection{Analysis of Calcium Recording Data} 
Here, we provide additional figures for the analysis of rat calcium imaging in \cite{reyesvallejo2023neurogliaform}. Table~\ref{tab:ClassifNeurons} shows the classification of 138 neurons from one site in mice brain. 
\begin{table}[h!]
    \centering
    \caption{Classification of neurons observed from a certain site of mice brain}
    \label{tab:ClassifNeurons}
    \begin{tabular}{@{}lrrrr@{}}
    \toprule
    Neuron classification & Adapting (AN) & Increased (IN) & Undefined-response (UN) & Total \\ \midrule
    Neurons in Layer 1               & 17 & 17 & 31 & 65    \\
    Neurons in Layer 2/3             & 21 & 12 & 40 & 73    \\
    Total Neurons                & 38 & 29 & 71 & 138   \\ \bottomrule
    \end{tabular}
    \end{table}

    
\end{appendix}

\newpage
\bibliographystyle{apalike} 
\bibliography{bibliography}       

\end{document}